\documentclass[submission,copyright,creativecommons,sharealike]{eptcs}
\providecommand{\event}{TiCSA 2023} %

\usepackage{iftex}

\ifpdf
  \usepackage{underscore}         %
  \usepackage[T1]{fontenc}        %
\else
  \usepackage{breakurl}           %
\fi

\usepackage[table]{xcolor}
\usepackage{xspace}
\usepackage{adjustbox}
\usepackage{amsmath}
\usepackage{amssymb} %

\definecolor{USPNcobalt}{HTML}{293358}
\definecolor{USPNocre}{HTML}{8b7d6d}
\definecolor{USPNblanc}{HTML}{ffffff}
\definecolor{USPNceruleen}{HTML}{354878}
\definecolor{USPNsable}{HTML}{ad947e}

\usepackage{varioref}
\usepackage[capitalise,english,nameinlink]{cleveref} %

\usepackage{paralist}
\newenvironment{inlineenumeration}
{\begin{inparaenum}[(\itshape i\upshape)]}
	{\end{inparaenum}}

\usepackage[shortlabels]{enumitem}

\usepackage{multirow}

\newcommand{\cellHeader}[0]{\cellcolor{USPNsable!40}\bfseries}
\newcommand{\rowHeader}{\rowcolor{USPNsable!40}\bfseries}

\newcommand{\cellYesOld}{\cellcolor{green!20}{$\surd$}}
\newcommand{\cellNoOld}{\cellcolor{red!20}{$\times$}}
\newcommand{\cellYesNew}{\cellcolor{green!35}\textbf{$\surd$}\bfseries}
\newcommand{\cellNoNew}{\cellcolor{red!35}\textbf{$\times$}\bfseries}
\newcommand{\cellOpen}{\cellcolor{yellow!20}\textbf{$?$}}
\newcommand{\cellCref}[1]{{\scriptsize (\cref{#1})}}

\usepackage{tikz}
\usetikzlibrary{arrows,automata,positioning,shapes} %

\tikzstyle{pta}=[auto, ->, >=stealth']
\tikzstyle{every node}=[initial text=]
\tikzstyle{location}=[rectangle, rounded corners, minimum size=12pt, draw=black, fill=blue!10, inner sep=2pt]
\tikzstyle{invariant}=[draw=black, dotted, inner sep=1pt, node distance=0] %
\tikzstyle{final}=[double, fill=blue!50]

\tikzstyle{private}=[fill=red,thick]

\definecolor{coloract}{rgb}{0.50, 0.70, 0.30}
\definecolor{colorclock}{rgb}{0.4, 0.4, 1}
\definecolor{colordisc}{rgb}{1, 0, 1}
\definecolor{colorloc}{rgb}{0.4, 0.4, 0.65}
\definecolor{colorparam}{rgb}{1, 0.6, 0.0}

\newcommand{\styleact}[1]{\ensuremath{\textcolor{coloract}{{\mathit{#1}}}}}
\newcommand{\styleclock}[1]{\ensuremath{\textcolor{colorclock}{{#1}}}}

\newcommand{\styleparam}[1]{\ensuremath{\textcolor{colorparam}{{#1}}}}

\newcommand{\textstyleact}[1]{\ensuremath{\mathit{#1}}}
\newcommand{\textstyleclock}[1]{\ensuremath{\mathit{#1}}}

\newcommand{\textstyleloc}[1]{\ensuremath{\mathrm{#1}}}
\newcommand{\textstyleparam}[1]{\ensuremath{{#1}}}

\newcommand{\defProblem}[3]
{%
	\noindent\fcolorbox{black}{USPNsable!20}{
	\begin{minipage}{.95\columnwidth}
		\textbf{#1:}\\
		\textsc{Input}: #2\\
		\textsc{Problem}: #3
	\end{minipage}
}

	\smallskip

}

\newcommand{\assign}{\leftarrow}

\newcommand{\checkUseMacro}[1]{#1}

\usepackage[english]{nomencl}
\usepackage{acro}

\newcommand{\stylecode}[1]{\textcolor{colorloc}{\texttt{#1}}}

\newcommand{\set}[1]{\ensuremath{\left\{#1\right\}}}

\newcommand{\setN}{\ensuremath{\mathbb{N}}}\nomenclature[M]{\setN}{Set of integers}
\newcommand{\setQ}{\ensuremath{\mathbb{Q}}}
\newcommand{\setQgeqzero}{\ensuremath{\setQ_{\geq 0}}}\nomenclature[M]{\setQgeqzero}{Set of non-negative rationals}
\newcommand{\setR}{\ensuremath{\mathbb{R}}}
\newcommand{\setRgeqzero}{\ensuremath{\setR_{\geq 0}}}\nomenclature[M]{\setRgeqzero}{Set of non-negative reals}
\newcommand{\setRplusinf}{\ensuremath{\setR^{\infty}_{\geq 0}}}\nomenclature[M]{\setRplusinf}{$\setRgeqzero\cup\set{+\infty}$}
\newcommand{\setZ}{\ensuremath{\mathbb{Z}}}\nomenclature[M]{\setZ}{Set of (positive and negative) integers}

\newcommand{\compOp}{\bowtie}

\newcommand{\intpart}[1]{\ensuremath{\lfloor#1\rfloor}}\nomenclature[M]{\ensuremath{\intpart{\bullet}}}{Integral part}
\newcommand{\fract}[1]{\ensuremath{\text{fr}(#1)}}\nomenclature[M]{\ensuremath{\fract{\bullet}}}{Fractionnal part}

\newcommand{\deriv}[1]{\ensuremath{\dot{#1}}}\nomenclature[M]{\ensuremath{\deriv{\bullet}}}{Derivative}
\newcommand{\wordOmega}[1]{\ensuremath{{#1}^\omega}}\nomenclature[M]{\ensuremath{\wordOmega{\bullet}}}{Repetition of zero or more occurrences in a word}

\newcommand{\init}{\ensuremath{0}}
\newcommand{\priv}{\ensuremath{{\mathit{priv}}}}
\newcommand{\pub}{\ensuremath{{\mathit{pub}}}}

\newcommand{\final}{\ensuremath{f}}
\newcommand{\valuate}[2]{\ensuremath{#2(#1)}}

\newcommand{\styleAutomaton}[1]{\ensuremath{\mathcal{#1}}}

\newcommand{\clock}{\ensuremath{\textstyleclock{x}}}\newcommand{\sclock}{\styleclock{\clock}}\nomenclature[C]{\clock}{A clock} %

\newcommand{\clocki}[1]{\ensuremath{\textstyleclock{\clock_{#1}}}}
\newcommand{\ClockCard}{H} %
\newcommand{\clockval}{\ensuremath{\mu}}\nomenclature[C]{\clockval}{Clock valuation} %
\newcommand{\ClockSet}{\ensuremath{\mathbb{X}}} %
\newcommand{\ClocksZero}{\ensuremath{\vec{0}}}\nomenclature[C]{\ClocksZero}{Clock valuation assigning $0$ to all clocks}

\newcommand{\resets}{\ensuremath{R}}
\newcommand{\reset}[2]{\ensuremath{[#1]_{#2}}}

\newcommand{\param}{\ensuremath{\textstyleparam{p}}}\newcommand{\sparam}{\styleparam{\param}}\nomenclature[C]{\param}{A parameter} %
\newcommand{\parami}[1]{\ensuremath{\textstyleparam{\param_{#1}}}} %
\newcommand{\paramiLower}[1]{\ensuremath{\textstyleparam{\parami{#1}^l}}}
\newcommand{\paramiUpper}[1]{\ensuremath{\textstyleparam{\parami{#1}^u}}}
\newcommand{\ParamCard}{\ensuremath{M}} %
\newcommand{\pval}{\ensuremath{v}} %
\newcommand{\PVal}{\ensuremath{V}}
\newcommand{\ParamSet}{\ensuremath{\mathbb{P}}} %

\newcommand{\TA}{\ensuremath{\checkUseMacro{\styleAutomaton{A}}}}\nomenclature[P]{\TA}{\TAtextdef}
\newcommand{\TAprivextend}{\ensuremath{\left(\ActionSet, \LocSet, \locinit, \locpriv, \locfinal, \ClockSet, \invariant, \EdgeSet\right)}}

\newcommand{\action}{\ensuremath{\textstyleact{a}}}\nomenclature[P]{\action}{An action}
\newcommand{\ActionSet}{\ensuremath{\Sigma}}\nomenclature[P]{\ActionSet}{Set of actions}

\newcommand{\constraint}{\ensuremath{C}}

\newcommand{\edge}{\ensuremath{\checkUseMacro{e}}}\nomenclature[P]{\edge}{An edge}
\newcommand{\edgei}[1]{\ensuremath{\checkUseMacro{\edge_{#1}}}}
\newcommand{\EdgeSet}{\ensuremath{E}}\nomenclature[P]{\EdgeSet}{Set of edges}

\newcommand{\guard}{\ensuremath{g}}

\newcommand{\invariant}{\ensuremath{I}} %

\newcommand{\loc}{\ensuremath{\textstyleloc{\ell}}}\nomenclature[P]{\loc}{A location} %
\newcommand{\loci}[1]{\ensuremath{\textstyleloc{\loc_{#1}}}}
\newcommand{\locinit}{\ensuremath{\textstyleloc{\loc_\init}}}
\newcommand{\locfinal}{\ensuremath{\textstyleloc{\loc_\final}}}
\newcommand{\locpriv}{\ensuremath{\textstyleloc{\loc_\priv}}}
\newcommand{\locpub}{\ensuremath{\textstyleloc{\loc_\pub}}}

\newcommand{\LocSet}{\ensuremath{L}}\nomenclature[P]{\LocSet}{Set of locations} %
\newcommand{\longuefleche}[1]{\stackrel{#1}{\longrightarrow}}

\newcommand{\run}{\checkUseMacro{\rho}} %

\newcommand{\duration}{\ensuremath{\mathit{dur}}} %
\newcommand{\runduration}[1]{\ensuremath{\checkUseMacro{\duration}(#1)}}
\newcommand{\statesDuration}[3]{\ensuremath{\checkUseMacro{\duration}_{#1}(#2,#3)}}
\newcommand{\locationsDuration}[3]{\ensuremath{\checkUseMacro{\duration}_{#1}(#2,#3)}}

\newcommand{\LUPTA}{\ensuremath{\styleAutomaton{P_{LU}}}}\nomenclature[P]{\LUPTA}{\LUPTAtextdef{}}
\newcommand{\PTA}{\ensuremath{\checkUseMacro{\styleAutomaton{P}}}}\nomenclature[P]{\PTAtext{}}{\PTAtextdef{}}
\newcommand{\PTAprivextend}{\ensuremath{\left(\ActionSet, \LocSet, \locinit,\locpriv, \locfinal, \ClockSet, \ParamSet, \invariant, \EdgeSet\right)}}

\newcommand{\styleSymbStatesSet}[1]{\ensuremath{\mathbf{#1}}}

\newcommand{\symbstate}{\ensuremath{\checkUseMacro{\styleSymbStatesSet{s}}}}\nomenclature[P]{\symbstate}{A symbolic state}

\newcommand{\SymbStateSet}{\ensuremath{\styleSymbStatesSet{S}}}\nomenclature[P]{\SymbStateSet}{Set of symbolic states}

\newcommand{\semantics}[1]{\ensuremath{\mathfrak{T}_{#1}}}
\newcommand{\semanticsextend}{\ensuremath{\left(\StateSet, \concstateinit, \ActionSet \cup \setRgeqzero, \transition\right)}}

\newcommand{\transition}{{\ensuremath{\rightarrow}}}
\newcommand{\transitionWith}[1]{\stackrel{#1}{\mapsto}}

\newcommand{\StateSet}{\ensuremath{\mathfrak{S}}}
\newcommand{\concstate}{\ensuremath{\mathfrak{s}}}
\newcommand{\concstatei}[1]{\ensuremath{\concstate_{#1}}}
\newcommand{\concstateinit}{\ensuremath{\concstate_\init}}

\newcommand{\bflag}{\ensuremath{b}}
\newcommand{\paramd}{\textstyleparam{\ensuremath{d}}}
\newcommand{\expiringBound}{\checkUseMacro{\ensuremath{\Delta}}}\nomenclature[P]{\expiringBound}{Expiring bound considered in \vref{sec:opacity:ICECCS}}
\newcommand{\setBound}{\ensuremath{\mathcal{D}}}\nomenclature[P]{\setBound}{Set of expiring bounds} %

\newcommand{\PrivVisit}[1]{\ensuremath{\mathit{Visit}^{\mathit{priv}}(#1)}}
\newcommand{\PubVisit}[1]{\ensuremath{\overline{\mathit{Visit}}^{\mathit{priv}}(#1)}}

\newcommand{\PrivDurVisit}[1]{\ensuremath{\mathit{DVisit}^\mathit{priv}(#1)}}
\newcommand{\PubDurVisit}[1]{\ensuremath{D\overline{\mathit{Visit}}^{\mathit{priv}}(#1)}}

\newcommand{\TempSupPrivVisit}[2]{\ensuremath{\mathit{Visit}^{\mathit{priv}}_{>#2}(#1)}}
\newcommand{\TempInfPrivVisit}[2]{\ensuremath{\mathit{Visit}^{\mathit{priv}}_{\leq#2}(#1)}}

\newcommand{\TempSupPrivDurVisit}[2]{\ensuremath{\mathit{DVisit}^\mathit{priv}_{>#2}(#1)}}
\newcommand{\TempInfPrivDurVisit}[2]{\ensuremath{\mathit{DVisit}^\mathit{priv}_{\leq#2}(#1)}}

\newcommand{\execTimeText}{\checkUseMacro{execution time}}
\newcommand{\execTimesText}{\checkUseMacro{execution times}}
\newcommand{\execTimes}{\ensuremath{D}}\nomenclature[P]{\execTimes}{Set of \execTimesText{}}

\newcommand{\imitator}{\textsf{IMITATOR}}

\newcommand{\textof}[1]{\ac{#1}} 					%
\newcommand{\textofPlural}[1]{\acp{#1}} 			%
\newcommand{\Textof}[1]{\Ac{#1}} 					%
\newcommand{\TextofPlural}[1]{\Acp{#1}} 			%

\newcommand{\titleSecTextOfPlural}[1]{\aclp*{#1}} 	%

\newcommand{\titleFigTextOf}[1]{\acs*{#1}} 			%
\newcommand{\titleFigTextOfPlural}[1]{\acsp*{#1}} 	%

\newcommand{\titleDefTextOf}[1]{\Acl*{#1}} 		%
\newcommand{\fullTextOf}[1]{\acf*{#1}} 			%
\DeclareAcronym{pta}{
	short=PTA,
	long=parametric timed automaton,
	short-plural=s,
	long-plural-form=parametric timed automata,
	cite=AHV93,
	extra={\vref{def:PTA}},
	tag={models}
}

\newcommand{\PTAtext}{\textof{pta}}
\newcommand{\PTAstext}{\textofPlural{pta}}

\newcommand{\PTAstextForSec}{\titleSecTextOfPlural{pta}}

\newcommand{\PTAtextForFig}{\acs*{pta}}

\newcommand{\PTAtextdef}{\titleDefTextOf{pta}}

\DeclareAcronym{lupta}{
	short=L/U-PTA,
	long=lower/upper parametric timed automaton,
	short-plural=s,
	long-plural-form=lower/upper parametric timed automata,
	short-indefinite={an},
	long-indefinite={a},
	cite=HRSV02,
	extra={\vref{def:LUPTA}},
	tag={models}
}

\newcommand{\LUPTAtext}{\iac{lupta}}
\newcommand{\LUPTAstext}{\textofPlural{lupta}}

\newcommand{\LUPTAstextForSec}{\titleSecTextOfPlural{lupta}}

\newcommand{\LUPTAstextForFig}{\titleFigTextOfPlural{lupta}}
\newcommand{\LUPTAtextdef}{\titleDefTextOf{lupta}}

\DeclareAcronym{upta}{
	short=U-PTA,
	long=upper-bound parametric timed automaton,
	short-plural=s,
	long-plural-form=upper-bound parametric timed automata,
	short-indefinite={a},
	long-indefinite={an},
	cite=BlT09,
	tag={models}
}
\newcommand{\UPTAtext}{\iac{upta}}
\newcommand{\UPTAstext}{\acp{upta}}

\DeclareAcronym{ta}{
	short=TA,
	long=timed automaton,
	short-plural=s,
	long-plural-form=timed automata,
	cite=AD94,
	extra={\vref{def:TA}},
	tag={models}
}

\newcommand{\TAtext}{\textof{ta}}
\newcommand{\TAstext}{\textofPlural{ta}}

\newcommand{\TAstextUpper}{\TextofPlural{ta}}

\newcommand{\TAtextForFig}{\titleFigTextOf{ta}}

\newcommand{\TAtextdef}{\titleDefTextOf{ta}}

\DeclareAcronym{ppta}{
	short=(P)TA,
	long=(possibly parametric) \acs*{ta},
	short-plural=s,
	long-plural-form=(possibly parametric) \acsp*{ta},
}

\DeclareAcronym{ipta}{
	short=IPTA,
	long=\imitator{} \ac*{pta},
	short-plural=s,
	long-plural-form=\imitator{} \acp*{pta},
	cite=Andre21,
	extra={\vref{sec:library:pta}},
	tag={models}
}

\DeclareAcronym{tga}{
	short=TGA,
	long=timed game automaton,
	short-plural=s,
	long-plural-form=timed game automata,
	cite=MPS95,
	tag={models}
}
\DeclareAcronym{pga}{
	short=PGA,
	long=parametric timed game automaton,
	short-plural=s,
	long-plural-form=parametric timed game automata,
	cite=JLR19,
	tag={models}
}
\DeclareAcronym{probTA}{
	short=PrTA,
	long=probabilistic timed automaton,
	short-plural=s,
	long-plural-form=probabilistic timed automata,
	cite=Beauquier03,
	tag={models}
}

\DeclareAcronym{lts}{
	short=LTS,
	long=labeled transition system,
	short-plural=s,
	long-plural-form=labeled transition systems,
	short-indefinite={an},
	long-indefinite={a},
	cite=Keller76,
	extra={\vref{def:LTS}},
	tag={models}
}

\DeclareAcronym{dfa}{
	short=DFA,
	long=deterministic finite-state automaton,
	short-plural=s,
	long-plural-form=deterministric finite-state automata,
	extra={\vref{def:DFA}},
	tag={models}
}

\DeclareAcronym{tts}{
	short=TTS,
	long=timed transition system,
	short-plural=s,
	long-plural-form=timed transition systems,
	cite=HMP91,
	extra={\vref{def:TTS}},
	tag={models}
}
\newcommand{\TTStext}{\textof{tts}}

\DeclareAcronym{era}{
	short=ERA,
	long=event-recording automaton,
	short-plural=s,
	long-plural-form=event-recording automata,
	cite=AFH99,
	tag={models}
}

\newcommand{\ERAstext}{\textofPlural{era}}

\DeclareAcronym{rta}{
	short=RTA,
	long=real-time automaton,
	short-plural=s,
	long-plural-form=real-time automata,
	tag={models}
}

\DeclareAcronym{pn}{
	short=PN,
	long=Petri net,
	short-plural=s,
	long-plural-form=Petri nets,
	cite=Pet62,
	tag={models}
}

\DeclareAcronym{ha}{
	short=HA,
	long=hybrid automaton,
	short-plural=s,
	long-plural-form=hybrid automata,
	cite=Henzinger96,
	tag={models}
}

\DeclareAcronym{mra}{
	short=MRA,
	long=multi-rate automaton,
	short-plural=s,
	long-plural-form=multi-rate automata,
	cite=ACHHHNOSY95,
	tag={models}
}

\DeclareAcronym{pzg}{
	short=PZG,
	long=parametric zone graph,
	extra={\vref{def:PTA:symbolic}},
	tag={misc}
}

\DeclareAcronym{opacity}{
	short=ET-opacity,
	long=execution-time opacity,
	tag={notion},
	extra={\vref{def:opacity:TOSEM:ET-opacity}},
	post={\acuse{opaque}}
}
\DeclareAcronym{opaque}{
	short=ET-opaque,
	long=execution-time opaque,
	tag={notion}, %
	post={\acuse{opacity}}
}

\newcommand{\opaqueText}{\textof{opaque}}
\newcommand{\opacityText}{\textof{opacity}}

\newcommand{\opaqueTextdef}{\fullTextOf{opaque}}

\newcommand{\opacityTextForSec}{\textof{opacity}}

\newcommand{\existentialOpaqueText}{\checkUseMacro{\ensuremath{\exists}-\acs*{opaque}}}
\newcommand{\existentialOpacityText}{\checkUseMacro{\ensuremath{\exists}-\acs*{opacity}}}
\newcommand{\weakOpaqueText}{\checkUseMacro{weakly \acs*{opaque}}}
\newcommand{\weakOpacityText}{\checkUseMacro{weak \acs*{opacity}}}
\newcommand{\fullOpaqueText}{\checkUseMacro{fully \acs*{opaque}}}
\newcommand{\fullOpacityText}{\checkUseMacro{full \acs*{opacity}}}

\newcommand{\WeakOpacityText}{\checkUseMacro{Weak \acs*{opacity}}}
\newcommand{\FullOpacityText}{\checkUseMacro{Full \acs*{opacity}}}

\newcommand{\existentialOpacityTextForSec}{\checkUseMacro{\ensuremath{\exists}-\acs*{opacity}}}
\newcommand{\weakOpacityTextForSec}{\checkUseMacro{weak \acs*{opacity}}}
\newcommand{\fullOpacityTextForSec}{\checkUseMacro{full \acs*{opacity}}}

\DeclareAcronym{tempopacity}{
	short=exp-\acs*{opacity},
	long=expiring \acl*{opacity},
	tag={notion},
	extra={\vref{def:opacity:ICECCS:temporary-timed-opacity}},
	post={\acuse{tempopaque}}
}
\DeclareAcronym{tempopaque}{
	short=exp-\acs*{opaque},
	long=expiring \acl*{opaque},
	tag={notion}, %
	post={\acuse{tempopacity}}
}

\newcommand{\tempOpaqueText}{\checkUseMacro{\textof{tempopaque}}}
\newcommand{\tempOpacityText}{\checkUseMacro{\textof{tempopacity}}}

\newcommand{\TempOpacityTextForSec}{\Textof{tempopacity}}

\newcommand{\tempOpaqueTextVal}[1]{\checkUseMacro{(\ensuremath{\leq}~\ensuremath{#1})-\acs*{opaque}}}
\newcommand{\tempOpacityTextVal}[1]{\checkUseMacro{(\ensuremath{\leq}~\ensuremath{#1})-\acs*{opacity}}}

\newcommand{\existentialTempOpaqueText}{\checkUseMacro{\ensuremath{\exists}-\acs*{tempopaque}}}
\newcommand{\existentialTempOpacityText}{\checkUseMacro{\ensuremath{\exists}-\acs*{tempopacity}}}
\newcommand{\weakTempOpaqueText}{\checkUseMacro{weakly \acs*{tempopaque}}}
\newcommand{\weakTempOpacityText}{\checkUseMacro{weak \acs*{tempopacity}}}
\newcommand{\fullTempOpaqueText}{\checkUseMacro{fully \acs*{tempopaque}}}
\newcommand{\fullTempOpacityText}{\checkUseMacro{full \acs*{tempopacity}}}
\newcommand{\weakfullTempOpaqueText}{\checkUseMacro{fully (resp.\ weakly) \acs*{tempopaque}}}
\newcommand{\weakfullTempOpacityText}{\checkUseMacro{full (resp.\ weak) \acs*{tempopacity}}}

\newcommand{\WeakfullTempOpacityText}{\checkUseMacro{Full (resp.\ weak) \acs*{tempopacity}}}

\newcommand{\existentialTempOpaqueTextVal}[1]{\checkUseMacro{\ensuremath{\exists}-\tempOpaqueTextVal{#1}}}

\newcommand{\weakTempOpaqueTextVal}[1]{\checkUseMacro{weakly \tempOpaqueTextVal{#1}}}
\newcommand{\weakTempOpacityTextVal}[1]{\checkUseMacro{weak \tempOpacityTextVal{#1}}}
\newcommand{\fullTempOpaqueTextVal}[1]{\checkUseMacro{fully \tempOpaqueTextVal{#1}}}
\newcommand{\fullTempOpacityTextVal}[1]{\checkUseMacro{full \tempOpacityTextVal{#1}}}
\newcommand{\weakfullTempOpaqueTextVal}[1]{\checkUseMacro{fully (resp.\ weakly) \tempOpaqueTextVal{#1}}}

\newcommand{\decisionProblem}[1]{#1 decision problem} %
\newcommand{\fullOpacityDecisionProblem}{\checkUseMacro{\decisionProblem{\fullOpacityText{}}}} %
\newcommand{\existentialOpacityDecisionProblem}{\checkUseMacro{\decisionProblem{\existentialOpacityText{}}}} %
\newcommand{\weakOpacityDecisionProblem}{\checkUseMacro{\decisionProblem{\weakOpacityText{}}}}

\newcommand{\FullOpacityDecisionProblem}{\checkUseMacro{\decisionProblem{\FullOpacityText{}}}} %
\newcommand{\WeakOpacityDecisionProblem}{\checkUseMacro{\decisionProblem{\WeakOpacityText{}}}}

\newcommand{\computationProblem}[1]{#1 t-computation problem} %
\newcommand{\opacityExecutionTimesComputationProblem}{\checkUseMacro{\computationProblem{\acs*{opacity}}}} %

\newcommand{\synthesisProblem}[2]{#1 #2-synthesis problem} %
\newcommand{\existentialOpacityParamSynthesisProblem}{\checkUseMacro{\synthesisProblem{\existentialOpacityText{}}{p}}} %
\newcommand{\fullOpacityParamSynthesisProblem}{\checkUseMacro{\synthesisProblem{\fullOpacityText{}}{p}}} %
\newcommand{\weakOpacityParamSynthesisProblem}{\checkUseMacro{\synthesisProblem{\weakOpacityText{}}{p}}} %

\newcommand{\FullOpacityParamSynthesisProblem}{\checkUseMacro{\synthesisProblem{\FullOpacityText{}}{p}}}
\newcommand{\WeakOpacityParamSynthesisProblem}{\checkUseMacro{\synthesisProblem{\WeakOpacityText{}}{p}}}

\newcommand{\emptinessProblemOneParam}[1]{#1-emptiness problem}
\newcommand{\emptinessProblem}[2]{\emptinessProblemOneParam{#1 #2}}
\newcommand{\existentialOpacityParamEmptinessProblem}{\checkUseMacro{\emptinessProblem{\existentialOpacityText{}}{p}}}
\newcommand{\fullOpacityParamEmptinessProblem}{\checkUseMacro{\emptinessProblem{\fullOpacityText{}}{p}}}
\newcommand{\weakOpacityParamEmptinessProblem}{\checkUseMacro{\emptinessProblem{\weakOpacityText{}}{p}}}

\newcommand{\FullOpacityParamEmptinessProblem}{\checkUseMacro{\emptinessProblem{\FullOpacityText{}}{p}}}
\newcommand{\WeakOpacityParamEmptinessProblem}{\checkUseMacro{\emptinessProblem{\WeakOpacityText{}}{p}}}

\newcommand{\weakTempDecisionProblem}{\decisionProblem{\weakTempOpacityText{}}}
\newcommand{\fullTempDecisionProblem}{\decisionProblem{\fullTempOpacityText{}}}
\newcommand{\weakfullTempDecisionProblem}{\decisionProblem{\weakfullTempOpacityText{}}}
\newcommand{\existentialTempDecisionProblem}{\decisionProblem{\existentialTempOpacityText{}}}

\newcommand{\WeakfullTempDecisionProblem}{\decisionProblem{\WeakfullTempOpacityText{}}}

\newcommand{\weakTempBoundEmptinessProblem}{\emptinessProblem{\weakTempOpacityText{}}{\expiringBound}}
\newcommand{\fullTempBoundEmptinessProblem}{\emptinessProblem{\fullTempOpacityText{}}{\expiringBound}}

\newcommand{\WeakfullTempBoundEmptinessProblem}{\emptinessProblem{\WeakfullTempOpacityText{}}{\expiringBound}}

\newcommand{\deltaAndP}{\ensuremath{\expiringBound}-p}

\newcommand{\weakTempBoundParametersEmptinessProblem}{\emptinessProblem{\weakTempOpacityText{}}{\deltaAndP}}

\newcommand{\weakfullTempBoundParametersEmptinessProblem}{\emptinessProblem{\weakfullTempOpacityText{}}{\deltaAndP}}

\newcommand{\WeakfullTempBoundParametersEmptinessProblem}{\emptinessProblem{\WeakfullTempOpacityText{}}{\deltaAndP}}

\newcommand{\fullTempBoundParametersSynthesisProblem}{\synthesisProblem{\fullTempOpacityText{}}{\deltaAndP}}
\newcommand{\weakfullTempBoundParametersSynthesisProblem}{\synthesisProblem{\weakfullTempOpacityText{}}{\deltaAndP}}

\newcommand{\WeakfullTempBoundParametersSynthesisProblem}{\synthesisProblem{\WeakfullTempOpacityText{}}{\deltaAndP}}

\newcommand{\boundComputationProblem}[1]{#1 \expiringBound-computation problem} %
\newcommand{\weakTempBoundComputationProblem}{\boundComputationProblem{\weakTempOpacityText{}}}
\newcommand{\fullTempBoundComputationProblem}{\boundComputationProblem{\fullTempOpacityText{}}}

\newcommand{\WeakfullTempBoundComputationProblem}{\boundComputationProblem{\WeakfullTempOpacityText{}}}

\newcommand{\expirationDateText}{expiration date}
\newcommand{\expirationDatesText}{expiration dates}

\newcommand{\ComplexityFont}[1]{{\sffamily\upshape #1}}

\newcommand{\FiveEXPTIME}{\ComplexityFont{5EXPTIME}\xspace}
\newcommand{\NEXPTIME}{\ComplexityFont{NEXPTIME}\xspace}
\newcommand{\PSPACE}{\ComplexityFont{PSPACE}\xspace}
\newcommand{\NLOGSPACE}{\ComplexityFont{NLOGSPACE}\xspace}

\definecolor{colorok}{RGB}{0,0,0}

\newcommand{\eg}{e.g.,\xspace}
\newcommand{\etal}{\emph{et al.}}
\newcommand{\ie}{i.e.,\xspace}
\newcommand{\st}{s.t.}

\newcommand{\wrt}{w.r.t.}

\usepackage{amsthm}
\theoremstyle{plain}

\newtheorem{proposition}{Proposition}
\newtheorem{theorem}{Theorem}
\newtheorem{corollary}{Corollary}

\theoremstyle{definition}
\newtheorem{definition}{Definition}
\newtheorem{example}{Example}

\theoremstyle{remark}
\newtheorem{remark}{Remark}

\title{Configuring Timing Parameters to Ensure Execution-Time Opacity in Timed Automata\thanks{This work is partially supported by the ANR-NRF French-Singaporean research program \href{https://www.loria.science/ProMiS/}{ProMiS} (ANR-19-CE25-0015 / 2019 ANR NRF 0092) and by ANR BisoUS (ANR-22-CE48-0012).}}

\author{Étienne André
\institute{Université Sorbonne Paris Nord, LIPN, CNRS UMR 7030\\F-93430 Villetaneuse, France}
\and
Engel Lefaucheux
\institute{Université de Lorraine, CNRS, Inria, LORIA\\
  F-54000 Nancy, France}
\and
Didier Lime
\institute{Nantes Université, École Centrale Nantes, CNRS, LS2N, UMR 6004\\F-44000 Nantes, France}
\and
Dylan Marinho
\institute{Université de Lorraine, CNRS, Inria, LORIA\\
  F-54000 Nancy, France}
\and
Jun Sun
\institute{School of Computing and Information Systems\\Singapore Management University}
}

\newcommand{\titlerunning}{Configuring Timing Parameters to Ensure Opacity in Timed Automata}
\newcommand{\authorrunning}{É.\ André, E.\ Lefaucheux, D.\ Lime, D.\ Marinho \& J.\ Sun}

\hypersetup{
  bookmarksnumbered,
  pdftitle    = {\titlerunning},
  pdfauthor   = {\authorrunning},
  pdfsubject  = {EPTCS},               %
  pdfkeywords = {parametric timed automata, opacity, cybersecurity} %
}

\begin{document}
\maketitle

\begin{abstract}
Timing information leakage occurs whenever an attacker successfully deduces confidential internal information by observing some timed information such as events with timestamps.
Timed automata are an extension of finite-state automata with a set of clocks evolving linearly and that can be tested or reset, making this formalism able to reason on systems involving concurrency and timing constraints.
In this paper, we summarize a recent line of works using timed automata as the input formalism, in which we assume that the attacker has access (only) to the system execution time.
First, we address the following execution-time opacity problem: given a timed system modeled by a timed automaton, given a secret location and a final location, synthesize the execution times from the initial location to the final location for which one cannot deduce whether the secret location was visited.
This means that for any such execution time, the system is opaque: either the final location is not reachable, or it is reachable with that execution time for both a run visiting and a run not visiting the secret location.
We also address the full execution-time opacity problem, asking whether the system is opaque for all execution times; we also study a weak counterpart.
Second, we add timing parameters, which are a way to configure a system: we identify a subclass of parametric timed automata with some decidability results. In addition, we devise a semi-algorithm for synthesizing timing parameter valuations guaranteeing that the resulting system is opaque.
Third, we report on problems when the secret has itself an expiration date, thus defining expiring execution-time opacity problems.
We finally show that our method can also apply to program analysis with configurable internal timings.
\end{abstract}

\section{Introduction}\label{section:introduction}
Complex timed systems often combine hard real-time constraints with concurrency. %
Information leakage, notably through side channels (see, \eg{} \cite{DKMR15,PBPMB17}), can have dramatic consequences on the security of such systems.
Among harmful information leaks, the \emph{timing information leakage} (see, \eg{} \cite{DXS18,GWW18,WGSW18,TCT19,VNNK19}) is the ability for an attacker to deduce internal information depending on observable timing information.
In this paper, we focus on timing leakage through the total execution time, \ie{} when a system works as an almost black-box and the ability of the attacker is limited to know the model and observe the total execution time.
We consider here the formalism of \TAstext{}, which is a popular extension of finite-state automata with clocks measuring time, \ie{} variables evolving linearly at the same rate.
Such clocks can be tested against integer constants in locations (``invariants'') or along transitions (``guards''), and can be reset to~$0$ when taking transitions.

\paragraph{Context and related works}
Franck~Cassez proposed in~\cite{Cassez09} a first definition of \emph{timed} opacity for \TAstext{}: the system is opaque if an attacker can never deduce whether some sequence of actions (possibly with timestamps) was performed, by only observing a given set of observable actions together with their timestamp.
It is then proved in~\cite{Cassez09} that it is undecidable whether a \TAtext{} is opaque, even for the restricted class of event-recording automata~\cite{AFH99} (a subclass of~\TAstext{}).
This notably relates to the undecidability of timed language inclusion for \TAstext{}~\cite{AD94}.
Security problems for \TAstext{} are surveyed in~\cite{AA23survey}.

The aforementioned negative result leaves hope only if the definition or the setting is changed, which was done in three main lines of works.
The different studied options were to reduce the expressiveness of the formalism~\cite{WZ18,WZA18}, to constrain the system to evolve in a time-bounded setting~\cite{AEYM21} or to consider a weaker attacker, who has access only to the \emph{execution time}~\cite{ALMS22,ALM23}, rather than to all observable actions with their timestamps.
We present here a summary of our recent works in this latter setting~\cite{ALMS22,ALM23}.

\paragraph{Contributions}
In the setting of \TAstext{}, we denote by \emph{execution time} the time from the system start to the time a given (final) location is entered.
Therefore, given a secret location, a \TAtext{} is \opaqueText{} for an execution time~$d$ if there exist at least two runs of duration~$d$ from the initial location to a final location: one visiting the secret location, and another one \emph{not} visiting the secret location.
In other words, if an attacker measures such an execution time from the initial location to the target location~$\locfinal$, then this attacker is not able to deduce whether the system visited~$\locpriv$.
Deciding whether at least one such~$d$ exists can be seen as an \emph{existential} version of \opacityText{} (called \existentialOpacityText{}).

Then, a \TAtext{} is \emph{\fullOpaqueText{}} if it is \opaqueText{} \emph{for all execution times}: that is, for each possible execution time~$d$, either the final location is unreachable, or the final location is reachable for at least two runs, one visiting the secret location, and another one not visiting it.
We define a \emph{weak} version of \opacityText{} by only requiring that runs visiting the secret location on the way to the final location have a counterpart of the same duration not visiting the secret location on the way to the final location, but not necessarily the opposite:
the \TAtext{} is \emph{weakly \opaqueText{}} if for each run visiting the secret location, there exists a run not visiting it with the same duration; the dual does not necessarily hold.

We also consider an \emph{expiring version} of \opacityText{}, where the secret is subject to an expiration date~$\expiringBound$.
That is, we consider that an attack is successful only when the attacker can decide that the secret location was visited less than $\expiringBound$ time units before the system completion.
Conversely, if the attacker exhibits an execution time~$d$ for which it is certain that the secret location was visited, but this location was visited strictly more than~$\expiringBound$ time units prior to the system completion, then this attack is useless, and can be seen as a failed attack.
The system is therefore \emph{fully expiring \opaqueText{}} %
 if the set of execution times for which the private location was visited within $\expiringBound$ time units prior to system completion (referred as ``secret times'') is exactly equal to the set of execution times for which the private location was either not visited or visited more than $\expiringBound$ time units prior to system completion (referred as ``non-secret times'').
Moreover, it is \emph{weakly expiring \opaqueText{}}
 when the inclusion of the secret times into the non-secret ones is verified---and not necessarily the dual.

Finally, we study the aforementioned problems for a \emph{parametric} extension of \TAstext{}, \ie{} \PTAstext{}, where integer constants compared to clocks can be made (rational-valued) timing parameters, \ie{} unknown constants.
Interesting problems include emptiness problems, \ie{} the emptiness of the parameter valuations set such that (expiring) \opacityText{} holds, and synthesis, \ie{} the synthesis of all parameter valuations such that (expiring) \opacityText{} holds.

\paragraph{About this manuscript}
This manuscript mainly summarizes results from two recent works, providing unified notations and concept names for the sake of consistency:
\begin{enumerate}
	\item defining and studying \opacityText{} problems~\cite{ALMS22} in \TAstext{} (\cref{section:opacity:TA}) and \PTAstext{} (\cref{section:opacity:PTA}); these notions from~\cite{ALMS22} are presented differently (including the problem names) in this paper for sake of consistency; and
	\item defining and studying \tempOpacityText{} problems~\cite{ALM23} in both \TAstext{} and \PTAstext{} (\cref{section:opacity:expiring}).
\end{enumerate}
In addition, we prove a few original results on \weakOpacityText{} (that were not addressed in~\cite{ALMS22} because we had not yet defined the concept of \weakOpacityText{} when writing~\cite{ALMS22}) and on \tempOpacityText{}.
These original results are \cref{prop:opacity:TOSEM:E-ET-decision,prop:opacity:TOSEM:W-ET-decision,prop:opacity:TOSEM:W-ET-pemptiness-PTA,prop:opacity:TOSEM:W-ET-pemptiness-LUPTA,thm:opacity:exist:TAtempop}.

In \cref{tab:resultsRecallTOSEM,tab:resultsRecallICECCS}, we summarize the decidability results recalled in this paper for \opacityText{} and \tempOpacityText{}.
We denote a problem with a green check if it is decidable, with a red cross if it is undecidable, and with a yellow question mark if it is open (or not considered in the aforementioned papers~\cite{ALMS22,ALM23}).
We emphasize using a bold font the original results of this paper.
The p-emptiness (resp.\ p-synthesis) problem asks for the synthesis (resp.\ for the non-existence) of a parameter valuation for which \opacityText{} is enforced.
The \deltaAndP{}-synthesis (resp.\ emptiness) problem asks for the synthesis (resp.\ for the non-existence) of a parameter valuation and an expiring bound $\expiringBound$ for which the \tempOpacityText{} is enforced.
L/U-PTA denote the lower-bound/upper-bound parametric timed automata~\cite{HRSV02} subclass of \PTAstext{}.
These notions will be formally defined in the paper.

\begin{table}[tb]
	\centering
	\caption{Summary of the results for \opacityText{}~\cite{ALMS22}}
	\begin{tabular}{|p{.2\textwidth} c | p{.2\textwidth} | p{.2\textwidth}| p{.2\textwidth}|}
		\hline
		\multicolumn{2}{|c|}{\cellHeader{}}                       & \cellHeader{\existentialOpaqueText{}}                        & \cellHeader{\weakOpaqueText{}} & \cellHeader{\fullOpaqueText{}}                               \\ \hline
		Decision                            & {\tiny \TAtext{}}   & \cellYesOld\cellCref{prop:opacity:TOSEM:E-ET-decision}           & \cellYesNew\cellCref{prop:opacity:TOSEM:W-ET-decision} & \cellYesOld \cellCref{prop:opacity:TOSEM:F-ET-decision}         \\ \hline
		\multirow{2}{*}{p-emptiness} & {\tiny \acs{lupta}} & \cellYesOld  \cellCref{prop:opacity:TOSEM:E-ET-pemptiness-LUPTA} & \cellNoNew  \cellCref{prop:opacity:TOSEM:W-ET-pemptiness-LUPTA} & \cellNoOld  \cellCref{prop:opacity:TOSEM:F-ET-pemptiness-LUPTA} \\
		& {\tiny \PTAtext{}}  & \cellNoOld \cellCref{prop:opacity:TOSEM:E-ET-pemptiness-PTA}     & \cellNoNew \cellCref{prop:opacity:TOSEM:W-ET-pemptiness-PTA} & \cellNoOld \cellCref{prop:opacity:TOSEM:F-ET-pemptiness-PTA}    \\ \hline
		\multirow{2}{*}{p-synthesis} & {\tiny \acs{lupta}} & \cellNoOld \cellCref{prop:opacity:TOSEM:E-ET-psynthesis-LUPTA}   & \cellNoNew\cellCref{prop:opacity:TOSEM:W-ET-psynthesis-LUPTA} & \cellNoOld\cellCref{thm:opacity:TOSEM:synthesisLUPTAfull}                                                      \\
		& {\tiny \PTAtext{}}  & \cellNoOld \cellCref{thm:opacity:TOSEM:synthesisPTA}                                                       & \cellNoNew\cellCref{prop:opacity:TOSEM:W-ET-psynthesis-PTA} & \cellNoOld\cellCref{thm:opacity:TOSEM:synthesisPTAfull}                                                      \\ \hline
	\end{tabular}
	\label{tab:resultsRecallTOSEM}
\end{table}
\begin{table}[tb]
	\centering
	\caption{Summary of the results for \tempOpacityText{}~\cite{ALM23}}
	\newcommand{\lengthFirstColumn}{.2\textwidth}
	\begin{tabular}{|p{\lengthFirstColumn} c | p{.2\textwidth} | p{.2\textwidth} | p{.2\textwidth} | }
		\hline
		\multicolumn{2}{|c|}{\cellHeader{}}                                                          & \cellHeader{\existentialTempOpaqueText{}} & \cellHeader{\weakTempOpaqueText{}}                      & \cellHeader{\fullTempOpaqueText{}}                         \\
\hline
		Decision                            & {\tiny \TAtext{}}   & \cellYesNew\cellCref{thm:opacity:exist:TAtempop}	           & \cellYesOld\cellCref{thm:opacity:ICECCS:TAtempop}	 & \cellYesOld \cellCref{thm:opacity:ICECCS:TAtempop}	         \\
		\hline
		$\expiringBound$-emptiness                                       & \multirow{2}{*}{\tiny TA} & \cellOpen                                 & \cellYesOld \cellCref{thm:opacity:ICECCS:weakEmptiness} & \cellYesOld \cellCref{thm:opacity:ICECCS:opacityEmptiness} \\
		$\expiringBound$-computation                                     &                           & \cellOpen                                 & \cellYesOld \cellCref{thm:opacity:ICECCS:weakcomp}      & \cellOpen                                                  \\ \hline
		\multirow{2}{\lengthFirstColumn}{\deltaAndP{}-emptiness} & {\tiny L/U-PTA}           & \cellOpen                                 & \cellNoOld \cellCref{thm:opacity:ICECCS:EmptinessLU}    & \cellNoOld \cellCref{thm:opacity:ICECCS:EmptinessLU}       \\
		                                                                 & {\tiny PTA}               & \cellOpen                                 & \cellNoOld \cellCref{thm:opacity:ICECCS:EmptinessPTA}   & \cellNoOld \cellCref{thm:opacity:ICECCS:EmptinessPTA}      \\ \hline
		\multirow{2}{\lengthFirstColumn}{\deltaAndP{}-synthesis} & {\tiny L/U-PTA}           & \cellOpen                                 & \cellNoOld  \cellCref{thm:opacity:ICECCS:synthesisLU}   & \cellNoOld \cellCref{thm:opacity:ICECCS:EmptinessLU}       \\
		                                                                 & {\tiny PTA}               & \cellOpen                                 & \cellNoOld \cellCref{thm:opacity:ICECCS:synthesisPTA}   & \cellNoOld \cellCref{thm:opacity:ICECCS:synthesisPTA}      \\ \hline
	\end{tabular}

	\label{tab:resultsRecallICECCS}
\end{table}

\paragraph{Outline}
\cref{section:preliminaries} recalls the necessary preliminaries, notably (parametric) timed automata.
\cref{section:opacity:TA} defines and reviews execution-time opacity problems in timed automata.
\cref{section:opacity:PTA} defines and reviews execution-time opacity problems in timed automata.
\cref{section:opacity:expiring} defines and reviews \emph{expiring} execution-time opacity problems in (parametric) timed automata.
\cref{section:implementation} briefly reports on our existing implementation of some of the problems using the parametric timed model checker \imitator{}~\cite{Andre21}.
\cref{section:conclusion} concludes the paper and reports on perspectives.

\section{Preliminaries}\label{section:preliminaries}

We denote by $\setN, \setZ, \setQgeqzero, \setRgeqzero$ the sets of non-negative integers, integers, non-negative rationals and non-negative reals, respectively.

\subsection{Clocks, parameters and constraints}

\emph{Clocks} are real-valued variables that all evolve over time at the same rate.
Throughout this paper, we assume a set~$\ClockSet = \{ \clocki{1}, \dots, \clocki{\ClockCard} \} $ of \emph{clocks}.
A \emph{clock valuation} is a function
$\clockval : \ClockSet \rightarrow \setRgeqzero$, assigning a non-negative value to each clock.
We write $\ClocksZero$ for the clock valuation assigning $0$ to all clocks.
Given a constant $d \in \setRgeqzero$, $\clockval + d$ denotes the valuation \st\ $(\clockval + d)(\clock) = \clockval(\clock) + d$, for all $\clock \in \ClockSet$.

A \emph{(timing) parameter} is an unknown rational-valued constant of a model.
Throughout this paper, we assume a set~$\ParamSet = \{ \parami{1}, \dots, \parami{\ParamCard} \} $ of \emph{parameters}.
A \emph{parameter valuation} $\pval$ is a function $\pval : \ParamSet \rightarrow \setQgeqzero$.

As often, we choose \emph{real}-valued clocks and \emph{rational}-valued parameters, because irrational constants render reachability undecidable in~\TAstext{}~\cite{Miller00} (see~\cite{Andre19STTT} for a survey on the impact of these domains in (P)TAs).

We assume ${\compOp} \in \{<, \leq, =, \geq, >\}$.
A constraint~$\constraint$ is a conjunction of inequalities over $\ClockSet \cup \ParamSet$ of the form
$\clock \compOp \sum_{1 \leq i \leq \ParamCard} \alpha_i \parami{i} + d$, with
$\parami{i} \in \ParamSet$,
and
$\alpha_i, d \in \setZ$.
Given~$\constraint$, we write~$\clockval\models\pval(\constraint)$ if %
the expression obtained by replacing each~$\clock$ with~$\clockval(\clock)$ and each~$\param$ with~$\pval(\param)$ in~$\constraint$ evaluates to true.

\subsection{Timed automata}

A \TAtext{} is a finite-state automaton extended with a finite set of real-valued clocks. %
We also add to the standard definition of \TAstext{} a special private location, which will be used to define our subsequent opacity concepts.

\begin{definition}[\TAtextdef{}~\cite{AD94}]
	A \TAtext{} $\TA$ is a tuple \mbox{$\TA = \TAprivextend$}, where:
	\begin{enumerate}
		\item $\ActionSet$ is a finite set of actions,
		\item $\LocSet$ is a finite set of locations,
		\item $\locinit \in \LocSet$ is the initial location,
		\item $\locpriv \in \LocSet$ is a special private location,
		\item $\locfinal \in \LocSet$ is the final location,
		\item $\ClockSet$ is a finite set of clocks,
		\item $\invariant$ is the invariant, assigning to every $\loc\in \LocSet$ a constraint $\invariant(\loc)$ over $\ClockSet$ (called \emph{invariant}),
		\item $\EdgeSet$ is a finite set of edges  $\edge = (\loc,\guard,\action,\resets,\loc')$
		where~$\loc,\loc'\in \LocSet$ are the source and target locations, $\action \in \ActionSet$,
		$\resets\subseteq \ClockSet$ is a set of clocks to be reset, and $\guard$ is a constraint over $\ClockSet$ (called \emph{guard}).
	\end{enumerate}
	\label{def:TA}
\end{definition}
\begin{figure}[tb]

	\centering

	\begin{tikzpicture}[pta, scale=2, xscale=1, yscale=.5]

		\node[location, initial] at (0, -1) (s0) {$\loci{0}$};

		\node[location, private] at (1.5, 0) (s2) {$\loci{2}$};

		\node[location, final] at (3, -1) (s1) {$\loci{1}$};

		\node[invariant, below=of s0] {$\styleclock{\clock} \leq 3$};
		\node[invariant, above=of s2] {$\styleclock{\clock} \leq 2$};

		\path (s0) edge[bend left] node[align=center]{$\styleact{a}$\\$\styleclock{\clock} \geq 1$} (s2);
		\path (s0) edge[] node[align=center]{$\styleact{c}$} (s1);
		\path (s2) edge[bend left] node[align=center]{$\styleact{b}$} (s1);

	\end{tikzpicture}
	\caption{A \TAtextForFig{} example}
	\label{figure:example-TA}

\end{figure}
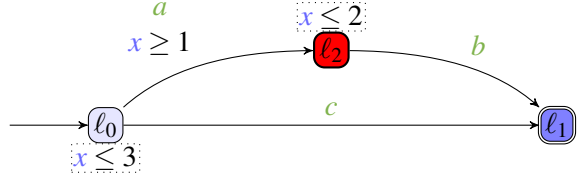
\begin{example}
	In \cref{figure:example-TA}, we give an example of a \TAtext{} with three locations $\loci{0}$, $\loci{1}$ and $\loci{2}$,
	three edges, three actions $\set{a, b, c}$,
	and one clock $\clock$.
	$\loci{0}$ is the initial location, $\loci{2}$ is the private location, while~$\loci{1}$ is the final location.
	$\loci{0}$ has an invariant $\styleclock{\clock} \leq 3$ and the edge from $\loci{0}$ to $\loci{2}$ has a guard $\styleclock{\clock} \geq 1$.
\end{example}
\paragraph{Concrete semantics of timed automata}

We recall the concrete semantics of a \TAtext{} using a \TTStext{}.

\begin{definition}[Semantics of a \TAtextForFig{}]
	Given a \TAtext{} $\TA = \TAprivextend$,
	the semantics of $\TA$ is given by the \TTStext{} $\semantics{\TA} = \semanticsextend$, with
	\begin{enumerate}
		\item $\StateSet = \{ (\loc, \clockval) \in \LocSet \times \setRgeqzero^\ClockCard \mid \clockval \models \invariant(\loc){\pval} \}$,
		\item $\concstateinit = (\locinit, \ClocksZero) $,
		\item  $\transition$ consists of the discrete and (continuous) delay transition relations:
		\begin{enumerate}
			\item discrete transitions: $(\loc,\clockval) \transitionWith{\edge} (\loc',\clockval')$,
			if $(\loc, \clockval) , (\loc',\clockval') \in \StateSet$, and there exists $\edge = (\loc,\guard,\action,\resets,\loc') \in \EdgeSet$, such that $\clockval'= \reset{\clockval}{\resets}$, and $\clockval\models\pval(\guard$).
			\item delay transitions: $(\loc,\clockval) \transitionWith{d} (\loc, \clockval+d)$, with $d \in \setRgeqzero$, if $\forall d' \in [0, d], (\loc, \clockval+d') \in \StateSet$.
		\end{enumerate}
	\end{enumerate}
\end{definition}

Moreover we write $(\loc, \clockval)\longuefleche{(d, \edge)} (\loc',\clockval')$ for a combination of a delay and discrete transition if
$\exists  \clockval'' :  (\loc,\clockval) \transitionWith{d} (\loc,\clockval'') \transitionWith{\edge} (\loc',\clockval')$.

Given a TA~$\TA$ with concrete semantics $\semanticsextend$, we refer to the states of~$\StateSet$ as the \emph{concrete states} of~$\TA$.
A \emph{run} of~$\TA$ is an alternating sequence of concrete states of~$\TA$ and pairs of edges and delays starting from the initial state $\concstateinit$ of the form
$(\loci{0}, \clockval_{0}), (d_0, \edge_0), (\loci{1}, \clockval_{1}), \cdots$
with
$i = 0, 1, \dots$, $\edge_i \in \EdgeSet$, $d_i \in \setRgeqzero$ and
$(\loci{i}, \clockval_{i}) \longuefleche{(d_i, \edge_i)} (\loci{i+1}, \clockval_{i+1})$.

\begin{definition}[Duration of a run]\label{definition:duration}
	Given a finite run $\run : (\loci{0}, \clockval_{0}), (d_0, \edgei{0}), (\loci{1}, \clockval_{1}), \cdots, (d_{i-1}, \edgei{i-1}), (\loci{n}, \clockval_{n})$, the \emph{duration} of $\run$ is $\runduration{\run} = \sum_{0 \leq i \leq n-1} d_i$.
	We also say that $\loci{n}$ is reachable in time~$\runduration{\run}$.
\end{definition}
\begin{example}
	Consider again the \TAtext{} \TA{} in \cref{figure:example-TA}.
	Consider the following run~$\run$ of $\TA$:
	$(\loci{0}, \clock = 0) , (1.4, a) , (\loci{2}, \clock = 1.4) , (0.4, b) , (\loci{1}, \clock = 1.8)$ %
	Note that we write ``$\clock = 1.4$'' instead of ``$\clockval$ such that $\clockval(\clock) = 1.4$''.
	We have $\runduration{\run} = 1.4 + 0.4 = 1.8$.
\end{example}
\subsection{Parametric timed automata}
A \PTAtext{} is a \TAtext{} extended with a finite set of timing parameters allowing to model unknown constants.

\begin{definition}[\PTAtextdef{}~\cite{AHV93}]
	A \PTAtext{} $\PTA$ is a tuple \mbox{$\PTA = \PTAprivextend$}, where:
	\begin{enumerate}
		\item $\ActionSet$ is a finite set of actions;
		\item $\LocSet$ is a finite set of locations;
		\item $\locinit \in \LocSet$ is the initial location;
		\item $\locpriv \in \LocSet$ is a special private location,
		\item $\locfinal \in \LocSet$ is the final location;
		\item $\ClockSet$ is a finite set of clocks;
		\item $\ParamSet$ is a finite set of parameters;
		\item $\invariant$ is the invariant, assigning to every $\loc\in \LocSet$ a constraint $\invariant(\loc)$ over $\ClockSet \cup \ParamSet$ (called \emph{invariant});
		\item $\EdgeSet$ is a finite set of edges  $\edge = (\loc,\guard,\action,\resets,\loc')$
		where~$\loc,\loc'\in \LocSet$ are the source and target locations, $\action \in \ActionSet$,
		$\resets\subseteq \ClockSet$ is a set of clocks to be reset, and $\guard$ is a constraint over $\ClockSet \cup \ParamSet$ (called \emph{guard}).
	\end{enumerate}
	\label{def:PTA}
\end{definition}
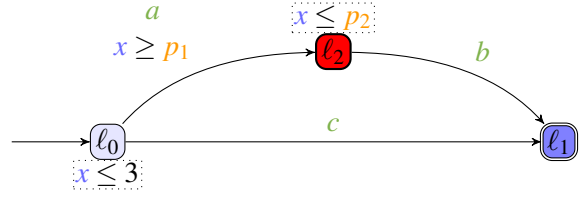
\begin{figure}[tb]
	\centering
	\begin{tikzpicture}[pta, scale=2, xscale=1, yscale=.6]

		\node[location, initial] at (0, -1) (s0) {$\loci{0}$};

		\node[location, private] at (1.5, 0) (s2) {$\loci{2}$};

		\node[location, final] at (3, -1) (s1) {$\loci{1}$};

		\node[invariant, below=of s0] {$\styleclock{\clock} \leq 3$};
		\node[invariant, above=of s2] {$\styleclock{\clock} \leq \styleparam{\parami{2}}$};

		\path (s0) edge[bend left] node[align=center]{$\styleact{a}$\\$\styleclock{\clock} \geq \styleparam{\parami{1}}$} (s2);
		\path (s0) edge[] node[align=center]{$\styleact{c}$} (s1);
		\path (s2) edge[bend left] node[align=center]{$\styleact{b}$} (s1);

	\end{tikzpicture}
	\caption{A \PTAtextForFig{} example}
	\label{figure:example-PTA}
\end{figure}
\begin{example}
	In \cref{figure:example-PTA}, we give an example of a \PTAtext{} with three locations $\loci{0}$, $\loci{1}$ and $\loci{2}$,
	three edges, three actions $\set{a, b, c}$,
	one clock~$\clock$
	and two parameters $\set{\parami{1}, \parami{2}}$.
	$\loci{0}$ is the initial location, $\loci{2}$ is the private location, while~$\loci{1}$ is the final location.
	$\loci{0}$ has an invariant $\styleclock{\clock} \leq 3$ and the edge from $\loci{0}$ to $\loci{2}$ has a guard $\styleclock{\clock} \geq \parami{1}$.
\end{example}
\begin{definition}[Valuation of a \PTAtextForFig{}]\label{def:valuation-PTA}
	Given a parameter valuation~$\pval$, we denote by $\valuate{\PTA}{\pval}$ the non-parametric structure where all occurrences of a parameter~$\parami{i}$ have been replaced by~$\pval(\parami{i})$.
\end{definition}
\begin{remark}
	We have a direct correspondence between the valuation of a \PTAtext{} and the definition of a \TAtext{} given in \cref{def:TA}.
	\TAstextUpper{} were originally defined with integer constants in \cite{AD94} (as done in \cref{def:TA}), while our definition of \PTAstext{} allows \emph{rational}-valued constants.
	By assuming a rescaling of the constants (\ie{} by multiplying all constants in a \TAtext{} by the least common multiple of their denominators), we obtain an equivalent (integer-valued) \TAtext{}, as defined in \cref{def:TA}.
	So we assume in the following that $\valuate{\PTA}{\pval}$ is a~\TAtext{}.
\end{remark}
\begin{example}
		Consider again the \PTAtext{} in \cref{figure:example-PTA} and let $\pval$ be such that $\pval(\parami{1}) = 1$ and $\pval(\parami{2}) = 2$.
		Then
		$\valuate{\PTA}{\pval}$ is the \TAtext{} depicted in \cref{figure:example-TA}.
\end{example}
\paragraph{\LUPTAtextdef{}}\label{sec:LUPTA-definition}

While most decision problems are undecidable for the general class of \PTAstext{} (see~\cite{Andre19STTT} for a survey),
\LUPTAstext{} is the most well-known subclass of \PTAstext{} with some decidability results: for example, reachability-emptiness (``the emptiness of the valuations set for which a given location is reachable''), which is undecidable for \PTAstext{}~\cite{AHV93}, becomes decidable for \LUPTAstext{}~\cite{HRSV02}.
Various other results were studied for this subclass (\eg{} \cite{BlT09,JLR15,ALR22}).

\begin{definition}[\LUPTAtextdef{}~\cite{HRSV02}]\label{def:LUPTA}
	\Iac{lupta} is a \PTAtext{} where the set of parameters is partitioned into lower-bound parameters and upper-bound parameters,
	where each upper-bound (resp.\ lower-bound) parameter~$\parami{i}$ must be such that,
	for every guard or invariant constraint $\clock \compOp \sum_{1 \leq i \leq \ParamCard} \alpha_i \parami{i} + d$, we have:
	\begin{itemize}
		\item ${\compOp} \in \set{ \leq, < }$ implies $\alpha_i \geq 0$ (resp.\ $\alpha_i \leq 0$), and
		\item ${\compOp} \in \set{ \geq, > }$ implies $\alpha_i \leq 0$ (resp.\ $\alpha_i \geq 0$).
	\end{itemize}
\end{definition}
\begin{example}
	The \PTAtext{} in \cref{figure:example-PTA} is \LUPTAtext{} with $\set{ \parami{1} }$ as lower-bound parameter set, and $\set{ \parami{2} }$ as upper-bound parameter set.
\end{example}
\section{Execution-time opacity problems in timed automata}\label{section:opacity:TA}

Throughout this paper, the attacker model is as follows:
the attacker knows the \TAtext{} modeling the system, and can only observe the \execTimeText{} between the start of the system and the time it reaches the final location.
The attacker cannot observe actions, nor the values of the clocks, nor whether some locations are visited.
Its goal will be to deduce from its observations whether the private location was visited.

\subsection{Defining the execution times}\label{sec:opacity:preliminaries:PrivPubVisit} %

Let us first introduce two key concepts necessary to define our notion of execution-time opacity.

Given a \TAtext{}~$\TA$ and a run~$\run$, we say that $\locpriv$ is \emph{visited on the way to~$\locfinal$ in~$\run$} if $\run$ is of the form
\[(\loci{0}, \clockval_0), (d_0, \edgei{0}), (\loci{1}, \clockval_1), \cdots, (\loci{m}, \clockval_m), (d_m, \edgei{m}), \cdots (\loci{n}, \clockval_n)\]
\noindent{}for some~$m,n \in \setN$ such that $\loci{m} = \locpriv$, $\loci{n} = \locfinal$ and $\forall 0 \leq i \leq n-1, \loci{i} \neq \locfinal$.
We denote by $\PrivVisit{\TA}$ the set of those runs, and refer to them as \emph{private} runs.
We denote by $\PrivDurVisit{\TA}$ the set of all the durations of these runs.

Conversely, we say that
$\locpriv$ is \emph{avoided on the way to~$\locfinal$ in~$\run$}
if $\run$ is of the form
\[(\loci{0}, \clockval_0), (d_0, \edgei{0}), (\loci{1}, \clockval_1), \cdots, (\loci{n}, \clockval_n )\]
\noindent{}with $\loci{n} = \locfinal$ and $\forall 0 \leq i < n, \loci{i} \notin \set{\locpriv,\locfinal}$.
We denote the set of those runs by~$\PubVisit{\TA}$, referring to them as \emph{public} runs,
and by $\PubDurVisit{\TA}$ the set of all the durations of these public runs.

Therefore, $\PrivDurVisit{\TA}$ (resp.\ $\PubDurVisit{\TA}$) is the set of all the durations of the runs for which $\locpriv$ is visited (resp.\ avoided) on the way to~$\locfinal$.

These concepts can be seen as the set of \execTimesText{} from the initial location~$\locinit$ to the final location $\locfinal$ while visiting (resp.\ not visiting) a private location~$\locpriv$.
Observe that, from the definition of the duration of a run (\cref{definition:duration}), this ``\execTimeText{}'' does not include the time spent in~$\locfinal$.

\begin{example}\label{example:opacity:TOSEM:running:DurVisit}
	Consider again the \TAtext{} in \cref{figure:example-TA}. %
	We have $\PrivDurVisit{\TA} = [1, 2]$
	and
	$\PubDurVisit{\TA} = [0, 3]$.
\end{example}
\subsection{Defining execution-time opacity}

We now introduce formally the concept of ``\opacityText{} for a set of durations (or \execTimesText{}) $\execTimes$'': a system is \emph{\opaqueText{} for \execTimesText{}~$\execTimes$} whenever, for any duration in~$\execTimes$, it is not possible to deduce whether the system visited~$\locpriv$ or not.
In other words, if an attacker measures an \execTimeText{} within~$\execTimes$ from the initial location to the target location~$\locfinal$, then this attacker is not able to deduce whether the system visited~$\locpriv$.

\begin{definition}[\Acf*{opacity} for $\execTimes$]\label{def:opacity:TOSEM:ET-opacity}
	Given a \TAtext{}~$\TA$ and a set of \execTimesText{}~$\execTimes$,
	we say that $\TA$ is \emph{\opaqueTextdef{} for \execTimesText{}~$\execTimes$}
	if $\execTimes \subseteq (\PrivDurVisit{\TA} \cap \PubDurVisit{\TA})$.
\end{definition}

In the following, we will be interested in the existence of such an \execTimeText{}.
We say that a \TAtext{} is \existentialOpaqueText{} if it is \opaqueText{} for a non-empty set of \execTimesText{}.

\begin{definition}[\existentialOpacityTextForSec{}]\label{def:opacity:TOSEM:exist-ET-opacity}
	A \TAtext{}~$\TA$ is \emph{\existentialOpaqueText{}}
	if $(\PrivDurVisit{\TA} \cap \PubDurVisit{\TA}) \neq \emptyset$.
\end{definition}

If one does not have the ability to tune the system (\ie{} change internal delays, or add some \stylecode{Thread.sleep()} statements in a program), one may be first interested in knowing whether the system is \opaqueText{} for all \execTimesText{}. %
In other words, if a system is \emph{\fullOpaqueText{}}, for any possible measured \execTimeText{}, an attacker is not able to deduce whether~$\locpriv$ was visited or not.

\begin{definition}[\fullOpacityTextForSec{}]\label{def:opacity:TOSEM:full-ET-opacity}
	A \TAtext{}~$\TA$ is \emph{\fullOpaqueText{}}
	if $\PrivDurVisit{\TA} = \PubDurVisit{\TA}$.
\end{definition}

That is, a system is \fullOpaqueText{} if, for any \execTimeText{} ~$\paramd$, a run of duration~$\paramd$ reaches~$\locfinal$ after visiting~$\locpriv$ iff another run of duration~$\paramd$ reaches~$\locfinal$ without visiting~$\locpriv$.

\begin{remark}\label{remark:opacity:TOSEM:R2}
	This definition is symmetric: a system is not \fullOpaqueText{} iff an attacker can deduce $\locpriv$ or $\neg \locpriv$.
	For instance, if there is no run to~$\locfinal$ visiting~$\locpriv$, but still a run to~$\locfinal$ (not visiting~$\locpriv$), a system is not \fullOpaqueText{} \wrt\ \cref{def:opacity:TOSEM:full-ET-opacity}.
\end{remark}

We finally define \weakOpacityText{}, not considered in~\cite{ALMS22}, but defined in the specific context of expiring opacity~\cite{ALM23}.
We therefore reintroduce this definition in the ``normal'' opacity setting considered in this section, in the following:

\begin{definition}[\weakOpacityTextForSec{}]\label{def:opacity:TOSEM:weak-ET-opacity}
	A \TAtext{}~$\TA$ is \emph{\weakOpaqueText{}}
	if $\PrivDurVisit{\TA} \subseteq \PubDurVisit{\TA}$.
\end{definition}

That is, a \TAtext{} is \weakOpaqueText{} whenever, for any run reaching the final location after visiting the private location, there exists another run of the same duration reaching the final location but not visiting the private location;
but the converse does not necessarily hold.

\begin{remark}\label{remark:opacity:ICECCS:weak}
	Our notion of \weakOpacityText{} may still leak some information:
	on the one hand, if a run indeed visits the private location, there exists an equivalent run not visiting it, and therefore the system is \opaqueText{};
	\emph{but} on the other hand, there may exist \execTimesText{} for which the attacker can deduce that the private location was \emph{not} visited.
	This remains acceptable in some cases, and this motivates us to define a weak version of \opacityText{}.
	Also note that the ``initial-state opacity'' for real-time automata considered in~\cite{WZ18} can also be seen as \emph{weak} in the sense that their language inclusion is also unidirectional.
\end{remark}

\begin{example}\label{example:opacity:TOSEM:defs}
	Consider again the \PTAtext{}~$\PTA$ in \cref{figure:example-PTA} and let $\pval$ such that $\pval(\parami{1}) = 1$ while  $\pval(\parami{2}) = 2$ (\ie{} the \TAtext{} in \cref{figure:example-TA}). %
	Recall that $\PrivDurVisit{\valuate{\PTA}{\pval}} = [1, 2]$
	and
	$\PubDurVisit{\valuate{\PTA}{\pval}} = [0, 3]$.
	Hence, it holds that $\PrivDurVisit{\valuate{\PTA}{\pval}} \subseteq \PubDurVisit{\valuate{\PTA}{\pval}}$ and therefore $\valuate{\PTA}{\pval}$ is \weakOpaqueText{}.
	However, $\PrivDurVisit{\valuate{\PTA}{\pval}} \neq \PubDurVisit{\valuate{\PTA}{\pval}}$ and therefore $\valuate{\PTA}{\pval}$ is not \fullOpaqueText{}.

Now consider again the \PTAtext{}~$\PTA$ in \cref{figure:example-PTA} and let $\pval'$ such that $\pval'(\parami{1}) = 0$ while  $\pval'(\parami{2}) = 3$.
This time, $\PrivDurVisit{\valuate{\PTA}{\pval'}} = \PubDurVisit{\valuate{\PTA}{\pval'}} = [0, 3]$ and therefore $\valuate{\PTA}{\pval'}$ is \fullOpaqueText{}.
\end{example}
\subsection{Decision and computation problems}\label{sec:opacity:TOSEM:problems}
\subsubsection{Computation problem for \opacityTextForSec{}}
We can now define the \opacityExecutionTimesComputationProblem{}, which consists in computing the possible \execTimesText{} ensuring \opacityText{}.

\defProblem
{\opacityExecutionTimesComputationProblem{}}
{A \TAtext{}~$\TA$
}
{Compute the \execTimesText{}~$\execTimes$ such that $\TA$ is \opaqueText{} for~$\execTimes$.}

Let us illustrate that this computation problem is certainly not easy.
For the \TAtext{}~$\TA$ in \cref{figure:opacity:TOSEM:opaqueforN}, the \execTimesText{}~$\execTimes$ for which $\TA$ is \opaqueText{} is exactly~$\setN$; that is, only integer times ensure \opacityText{} (as the system can only leave $\locpriv$ and hence enter~$\locfinal$ at an integer time), while non-integer times violate \opacityText{}.

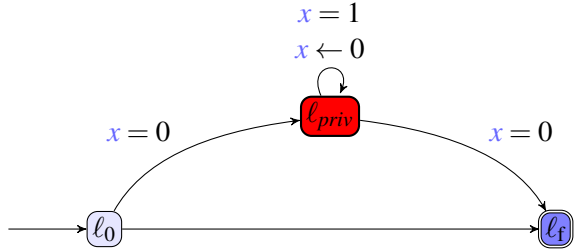
\begin{figure}[tb]
	{\centering
		\begin{tikzpicture}[->, >=stealth', auto, node distance=2cm, xscale=2, thin]

			\node[location, initial] at (0, 0) (s0) {$\locinit$};
			\node[location, private] at (1.5,1.5) (sp) {$\locpriv$};
			\node[location, final] at (3,0) (s1) {$\locfinal$};

			\path (s0) edge node[align=center]{} (s1);
			\path (s0) edge[bend left] node[align=center]{$\sclock = 0$} (sp);
			\path (sp) edge[loop above] node[align=center]{$\sclock = 1$\\$\sclock \assign 0$} (sp);
			\path (sp) edge[bend left] node[align=center]{$\sclock = 0$} (s1);

		\end{tikzpicture}

	}
	\caption{\TAtextForFig{} for which the set of \execTimesText{} ensuring \opacityTextForSec{} is $\setN$}
	\label{figure:opacity:TOSEM:opaqueforN}
\end{figure}
\subsubsection{Decision problems} %

We define the three following decision problems:

\defProblem{\existentialOpacityDecisionProblem{}}
{A \TAtext{}~$\TA$} %
{Is $\TA$ \existentialOpaqueText{}?}

\defProblem{\FullOpacityDecisionProblem{}}
{A \TAtext{}~$\TA$} %
{Is $\TA$ \fullOpaqueText{}?}

\defProblem{\WeakOpacityDecisionProblem{}}
{A \TAtext{}~$\TA$} %
{Is $\TA$ \weakOpaqueText{}?}

\subsection{Answering the \opacityExecutionTimesComputationProblem{}}\label{sec:opacity:TOSEM:ET-t-computation}
\begin{proposition}[Solvability of the \opacityExecutionTimesComputationProblem{}~{\cite[Proposition~5.2]{ALMS22}}]\label{prop:opacity:TOSEM:ET-t-computation}
The \opacityExecutionTimesComputationProblem{} is solvable for \TAstext{}.
\end{proposition}

This positive result can be put in perspective with the negative result of~\cite{Cassez09} that proves that it is undecidable whether a \TAtext{} (and even the more restricted subclass of \ERAstext{}) is opaque, in a sense that the attacker can deduce some actions, by looking at observable actions together with their timing.
The difference in our setting is that only the global time is observable, which can be seen as a single action, occurring once only at the end of the computation.
In other words, our attacker is less powerful than the attacker in~\cite{Cassez09}.

\subsection{Checking for \existentialOpacityTextForSec{}}\label{sec:opacity:TOSEM:E-ET-decision}

The following result was not strictly speaking proved in~\cite{ALMS22}, and we provide here an original proof for it.

\begin{proposition}[Decidability of the \existentialOpacityDecisionProblem{}]\label{prop:opacity:TOSEM:E-ET-decision}
The \existentialOpacityDecisionProblem{} is decidable in \FiveEXPTIME{} for \TAstext{}.
\end{proposition}
\begin{proof}
Let $\TA$~be a \TAtext{}.
Suppose we add a Boolean variable $\priv$ to $\TA$ which is initially false and set to true on every edge going into the location $\locpriv$.
This Boolean variable (not strictly part of the \TAtext{} syntax) can also be simulated by adding a copy of~$\TA$ instead, and jumping to that copy on edges going into location~$\locpriv$.

Then the \existentialOpacityDecisionProblem{} amounts to checking the following parametric TCTL formula~\cite{BDR08}, with $\param$ a parameter:
$$\exists \param \big(\exists\Diamond_{=\param} (\locfinal \wedge \priv) \wedge \exists\Diamond_{=\param} (\locfinal \wedge \neg \priv) \big)$$

From~\cite{BDR08}, this can be checked in \FiveEXPTIME{}, since the size of the \TAtext{} it is checked on is at most twice that of \TA{}, and the size of the formula is constant \wrt\ the size of \TA.
\end{proof}

\subsection{Checking for \fullOpacityTextForSec{}}\label{sec:opacity:TOSEM:F-ET-decision}

The following result matches \cite[Proposition~5.3]{ALMS22} but we provide an original proof, also fixing a complexity issue in \cite[Proposition~5.3]{ALMS22}.

\begin{proposition}[Decidability of the \fullOpacityDecisionProblem{}]\label{prop:opacity:TOSEM:F-ET-decision}
The \fullOpacityDecisionProblem{} is decidable in \FiveEXPTIME{} for \TAstext{}.
\end{proposition}
\begin{proof}
As before, we can write a parametric TCTL formula for this problem, with $\param$ a parameter:
$$\forall \param \big(\exists\Diamond_{=\param} (\locfinal \wedge \priv) \Leftrightarrow \exists\Diamond_{=\param} (\locfinal \wedge \neg \priv) \big)$$
This formula can be checked in \FiveEXPTIME{}~\cite{BDR08}.
\end{proof}
\subsection{Checking for \weakOpacityTextForSec{}}

The \emph{weak} notion of \opacityText{} had not been defined in~\cite{ALMS22}.
Nevertheless, the proof of \cref{prop:opacity:TOSEM:F-ET-decision} can be adapted in a very straightforward manner to prove its weak counterpart as follows:

\begin{proposition}[Decidability of the \weakOpacityDecisionProblem{}]\label{prop:opacity:TOSEM:W-ET-decision}
The \weakOpacityDecisionProblem{} is decidable in \FiveEXPTIME{} for \TAstext{}.
\end{proposition}
\begin{proof}
Let $\TA$~be a \TAtext{}.
As before, we can write a parametric TCTL formula for this problem, with $\param$ a parameter:
$$\forall \param \big(\exists\Diamond_{=\param} (\locfinal \wedge \priv) \Rightarrow \exists\Diamond_{=\param} (\locfinal \wedge \neg \priv) \big)$$
This formula can be checked in \FiveEXPTIME{}~\cite{BDR08}.
\end{proof}
\section{Execution-time opacity problems in parametric timed automata}\label{section:opacity:PTA}

We now extend opacity problems to parametric timed automata.
We first address the parametric problems related to \existentialOpacityText{} in \cref{sec:opacity:TOSEM:existential}.
The decision problems associated to \fullOpacityText{} and \weakOpacityText{} will then be considered in \cref{sec:opacity:TOSEM:parametric-F-ET,sec:opacity:TOSEM:parametric-W-ET} respectively.

Following the usual concepts for parametric timed automata, we consider both \emph{emptiness} and \emph{synthesis} problems.
An emptiness problem aims at \emph{deciding} whether the set of parameter valuations for which a given property holds in the valuated \TAtext{} is empty, while a synthesis problem aims at \emph{synthesizing} the set of parameter valuations for which a given property holds in the valuated \TAtext{}.

\subsection{\existentialOpacityText{}}\label{sec:opacity:TOSEM:existential}
\subsubsection{Problems}
\paragraph{Emptiness problem for \existentialOpacityTextForSec{}}

Let us consider the following decision problem, \ie{} the problem of checking the \emph{emptiness} of the set of parameter valuations guaranteeing \existentialOpacityText{}.

\defProblem
{\existentialOpacityParamEmptinessProblem}
{A \PTAtext{}~$\PTA$
}
{Decide the emptiness of the set of parameter valuations $\pval$
such that $\valuate{\PTA}{\pval}$ is \existentialOpaqueText{}.}

The negation of the \existentialOpacityParamEmptinessProblem{} consists in deciding whether there exists at least one parameter valuation for which $\valuate{\PTA}{\pval}$ is \existentialOpaqueText{}.

\paragraph{Synthesis problem for \existentialOpacityTextForSec{}}

The synthesis counterpart allows for a higher-level problem by also synthesizing the internal timings guaranteeing \existentialOpacityText{}. %

\defProblem
{\existentialOpacityParamSynthesisProblem{}}
{A \PTAtext{}~$\PTA$
}
{Synthesize the set~$\PVal$ of parameter valuations such that $\valuate{\PTA}{\pval}$ is \existentialOpaqueText{}, for all $\pval \in \PVal$.}

\subsubsection{Undecidability in general}

With the rule of thumb that all non-trivial decision problems are undecidable for general \PTAstext{}~\cite{Andre19STTT}, the following result is not surprising, and follows from the undecidability of reachability-emptiness for \PTAstext{}~\cite{AHV93}.

\begin{theorem}[Undecidability of the \existentialOpacityParamEmptinessProblem{}~{\cite[Theorem~6.1]{ALMS22}}]\label{prop:opacity:TOSEM:E-ET-pemptiness-PTA}
The \existentialOpacityParamEmptinessProblem{} is undecidable for general \PTAstext{}.
\end{theorem}

Since the emptiness problem is undecidable, the synthesis problem is immediately unsolvable as well.

\begin{corollary}\label{thm:opacity:TOSEM:synthesisPTA}
	The \existentialOpacityParamSynthesisProblem{} is unsolvable for general \PTAstext{}.
\end{corollary}

Nevertheless, in~\cite{ALMS22} we proposed a procedure solving this problem.
While this procedure is not guaranteed to terminate, its result is correct when termination can be achieved.
See~\cite[Section~8]{ALMS22} for details.

\subsubsection{The subclass of \LUPTAstext{}}\label{ss:ET-opacity-LU}

\paragraph{Decidability}

We now show that the \existentialOpacityParamEmptinessProblem{} is decidable for \LUPTAstext{}.
Despite early positive results for \LUPTAstext{}~\cite{HRSV02,BlT09}, more recent results (notably \cite{JLR15,ALR18FORMATS,ALR22}) mostly proved undecidable properties of \LUPTAstext{}, and therefore this positive result is welcome.

\begin{theorem}[Decidability of the \existentialOpacityParamEmptinessProblem{}~{\cite[Theorem~6.2]{ALMS22}}]\label{prop:opacity:TOSEM:E-ET-pemptiness-LUPTA}
The \existentialOpacityParamEmptinessProblem{} is decidable for \LUPTAstext{}.
\end{theorem}

\paragraph{Intractability of synthesis for \LUPTAstextForSec{}}

Even though the \existentialOpacityParamEmptinessProblem{} is decidable for \LUPTAstext{} (\cref{prop:opacity:TOSEM:E-ET-pemptiness-LUPTA}), the \emph{synthesis} of the parameter valuations remains intractable in general, as shown in the following \cref{prop:opacity:TOSEM:E-ET-psynthesis-LUPTA}.
By intractable we mean more precisely that the solution, if it can be computed, cannot (in general, \ie{} for some sufficiently complex solutions) be represented using any formalism for which the emptiness of the intersection with equality constraints is decidable.
That is, a formalism in which it is decidable to decide ``the emptiness of the valuation set of the computed solution intersected with an equality test between variables''\ cannot be used to represent the solution.
For example, let us question whether we could represent the solution of the \existentialOpacityParamSynthesisProblem{} for \LUPTAstext{} using the formalism of a \emph{finite union of polyhedra}: testing whether a finite union of polyhedra intersected with ``equality constraints'' (typically $\parami{1} = \parami{2}$) is empty or not \emph{is} decidable.
The Parma polyhedra library~\cite{BHZ08} can typically compute the answer to this question.
Therefore, from the following \cref{prop:opacity:TOSEM:E-ET-psynthesis-LUPTA}, finite unions of polyhedra cannot be used to represent the solution of the \existentialOpacityParamSynthesisProblem{} for \LUPTAstext{}.
As finite unions of polyhedra are a very common formalism (not to say the \emph{de facto} standard) to represent the solutions of various timing parameters synthesis problems, the synthesis is then considered to be infeasible in practice, or \emph{intractable} (following the vocabulary used in \cite[Theorem~2]{JLR15}).

\begin{proposition}[Intractability of the \existentialOpacityParamSynthesisProblem{}~{\cite[Proposition~6.4]{ALMS22}}]\label{prop:opacity:TOSEM:E-ET-psynthesis-LUPTA}
In case a solution to the  \existentialOpacityParamSynthesisProblem{} for \LUPTAstext{} can be computed, this solution may be not representable using any formalism for which the emptiness of the intersection with equality constraints is decidable.
\end{proposition}
\subsection{Parametric \fullOpacityTextForSec{}}\label{sec:opacity:TOSEM:parametric-F-ET}

We address here the following decision problem, which asks about the emptiness of the parameter valuation set guaranteeing \fullOpacityText{}.
We also define the \fullOpacityParamSynthesisProblem{}, this time \emph{synthesizing} the timing parameters guaranteeing \fullOpacityText{}.

\subsubsection{Problem definitions}

\defProblem
{\FullOpacityParamEmptinessProblem{}}
{A \PTAtext{}~$\PTA$
}
{Decide the emptiness of the set of parameter valuations $\pval$
such that $\valuate{\PTA}{\pval}$ is \fullOpaqueText{}.}

Equivalently, we are interested in deciding whether there exists at least one parameter valuation for which $\valuate{\PTA}{\pval}$ is \fullOpaqueText{}.

We also define the \emph{\fullOpacityParamSynthesisProblem{}}, aiming at synthesizing (ideally the entire set of) parameter valuations~$\pval$ for which $\valuate{\PTA}{\pval}$ is \fullOpaqueText{}.

\defProblem
{\FullOpacityParamSynthesisProblem{}}
{A \PTAtext{}~$\PTA$
}
{Synthesize the set~$\PVal$ of parameter valuations such that $\valuate{\PTA}{\pval}$ is \fullOpaqueText{}, for all $\pval \in \PVal$.}

\subsubsection{Undecidability for general PTAs}

Considering that \cref{prop:opacity:TOSEM:E-ET-pemptiness-PTA} shows the undecidability of the \existentialOpacityParamEmptinessProblem{}, the undecidability of the \fullOpacityParamEmptinessProblem{} is not surprising, but does not follow immediately.

\begin{theorem}[Undecidability of the \fullOpacityParamEmptinessProblem{}~{\cite[Theorem~7.2]{ALMS22}}]\label{prop:opacity:TOSEM:F-ET-pemptiness-PTA}
The \fullOpacityParamEmptinessProblem{} is undecidable for general \PTAstext{}.
\end{theorem}

The proof relies on a reduction from the problem of reachability-emptiness in constant time, a result proved itself undecidable in the same paper \cite[Lemma~7.1]{ALMS22}.

Since the emptiness problem is undecidable, the synthesis problem is immediately unsolvable as well.

\begin{corollary}\label{thm:opacity:TOSEM:synthesisPTAfull}
	The \fullOpacityParamSynthesisProblem{} is unsolvable for \PTAstext{}. %
\end{corollary}
\subsubsection{Undecidability for \LUPTAstextForSec{}}
\begin{figure}[tb]
{\centering
\begin{tikzpicture}[->, >=stealth', auto, node distance=2cm, thin]

	\node[location, initial] at (0, 0) (s0) {$\locinit$};
	\node[location, private] at (1.5,1.5) (sp) {$\locpriv$};
	\node[location, final] at (3,0) (s1) {$\locfinal$};

	\path (s0) edge node[align=center]{$\sclock \leq \sparam$} (s1);
	\path (s0) edge node[align=center]{} (sp);
	\path (sp) edge[] node[align=center]{$\sclock \leq 1$} (s1);

\end{tikzpicture}

}
\caption{No monotonicity for \fullOpacityTextForSec{} in \LUPTAstextForFig{}}
\label{figure:opacity:TOSEM:counterexample-monotonicity-LU}
\end{figure}
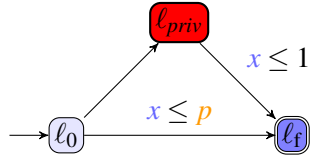

Let us now study the \fullOpacityParamEmptinessProblem{} for \LUPTAstext{}.
While it is well-known that \LUPTAstext{} enjoy a monotonicity for reachability properties (``enlarging an upper-bound parameter or decreasing a lower-bound parameter preserves reachability'')~\cite{HRSV02}, we can show in the following example that this is not the case for \fullOpacityText{}.

\begin{example}
Consider the \PTAtext{} in \cref{figure:opacity:TOSEM:counterexample-monotonicity-LU}.
First assume $\pval$ such that $\pval(\param) = 0.5$.
Then, $\valuate{\PTA}{\pval}$ is not \fullOpaqueText{}: indeed, $\locfinal$ can be reached in $1$ time unit by visiting~$\locpriv$, but not without visiting~$\locpriv$.

Second, assume $\pval'$ such that $\pval'(\param) = 1$.
Then, $\valuate{\PTA}{\pval'}$ is \fullOpaqueText{}: indeed, $\locfinal$ can be reached for any duration in $[0,1]$ by runs both visiting and not visiting~$\locpriv$.

Finally, let us enlarge~$\param$ further, and assume $\pval''$ such that $\pval''(\param) = 2$.
Then, $\valuate{\PTA}{\pval''}$ becomes again not \fullOpaqueText{}: indeed, $\locfinal$ can be reached in $2$ time units without visiting~$\locpriv$, but cannot be reached in $2$ time units by visiting~$\locpriv$.

As a side note, remark that this \PTAtext{} is actually \UPTAtext{}, that is, monotonicity for this problem does not even hold for \UPTAstext{}.
\end{example}

\smallskip

In fact, we show that, while the \existentialOpacityParamEmptinessProblem{} is decidable for \LUPTAstext{} (\cref{prop:opacity:TOSEM:E-ET-pemptiness-LUPTA}), the \fullOpacityParamEmptinessProblem{} becomes undecidable for this same class. %
This confirms (after previous works in~\cite{BlT09,JLR15,ALR18FORMATS,ALR22}) that \LUPTAstext{} stand at the frontier between decidability and undecidability.

\begin{theorem}[Undecidability of the \fullOpacityParamEmptinessProblem{} for \LUPTAstextForFig{}~{\cite[Theorem~7.4]{ALMS22}}]\label{prop:opacity:TOSEM:F-ET-pemptiness-LUPTA}
The \fullOpacityParamEmptinessProblem{} is undecidable for \LUPTAstext{}. %
\end{theorem}

Since the emptiness problem is undecidable, the synthesis problem is immediately unsolvable as well.

\begin{corollary}\label{thm:opacity:TOSEM:synthesisLUPTAfull}
	The \fullOpacityParamSynthesisProblem{} is unsolvable for \LUPTAstext{}. %
\end{corollary}
\begin{remark}\label{remark:redundancy}
	Since L/U-PTAs are a subclass of PTAs (put it differently: ``any L/U-PTA is a PTA''), the negative results proved for L/U-PTAs (\cref{prop:opacity:TOSEM:F-ET-pemptiness-LUPTA,thm:opacity:TOSEM:synthesisLUPTAfull}) immediately imply those previously shown for general PTAs (\cref{prop:opacity:TOSEM:F-ET-pemptiness-PTA,thm:opacity:TOSEM:synthesisPTAfull}).
	However, in~\cite{ALMS22}, a smaller number of clocks and parameters is needed to prove the aforementioned negative results for general PTAs, which justifies the two versions of the proofs in~\cite{ALMS22}.
\end{remark}
\subsection{Parametric \weakOpacityTextForSec{}}\label{sec:opacity:TOSEM:parametric-W-ET}
\subsubsection{Problem definitions}

\defProblem
{\WeakOpacityParamEmptinessProblem{}}
{A \PTAtext{}~$\PTA$}
{Decide the emptiness of the set of parameter valuations $\pval$ such that $\valuate{\PTA}{\pval}$ is \weakOpaqueText{}.}

\defProblem
{\WeakOpacityParamSynthesisProblem{}}
{A \PTAtext{}~$\PTA$}
{Synthesize the parameter valuations~$\pval$ such that $\valuate{\PTA}{\pval}$ is \weakOpaqueText{}.}

\subsubsection{Undecidability for general PTAs}

We provide below an original result in the context of \emph{weak} opacity, but partially inspired by the construction used in the proof of \cref{prop:opacity:TOSEM:F-ET-pemptiness-PTA}.

\begin{theorem}[Undecidability of the \weakOpacityParamEmptinessProblem{}]\label{prop:opacity:TOSEM:W-ET-pemptiness-PTA}
The \weakOpacityParamEmptinessProblem{} is undecidable for general \PTAstext{}.
\end{theorem}
\begin{figure}[tb]
{\centering
\begin{tikzpicture}[->, >=stealth', auto, node distance=2cm, thin]

	\node[location] (l0) at (-2, 0) {$\locinit$};
	\node[location] (lf) at (+1.8, 0) {$\locfinal$};
	\node[cloud, cloud puffs=15.7, cloud ignores aspect, minimum width=5cm, minimum height=2cm, align=center, draw] (cloud) at (0cm, 0cm) {$\PTA$};

	\node[location, initial] (l0') at (-5, 0) {$\locinit'$};
	\node[location] (lpriv) at (+5, 0) {$\locpub$};
	\node[location, private] (lpub) at (+2, -1.2) {$\locpriv$};
	\node[location, final] (lf') at (+5, -1.2) {$\locfinal'$};

	\path
	(l0') edge node {$\styleclock{\clock} = 0$} (l0) %
	(lf) edge node[xshift=.5em] {$\styleclock{\clock} = 1$} (lpriv)
	(lpriv) edge node {$\styleclock{\clock} = 1$} (lf')
	(l0') edge[out=-45,in=180] node {$\styleclock{\clock} = 0$} (lpub)
	(lpub) edge node {$\styleclock{\clock} = 1$} (lf')
	;

\end{tikzpicture}

}
\caption{Reduction from reachability-emptiness for the proof of \cref{prop:opacity:TOSEM:W-ET-pemptiness-PTA}} %
\label{figure:opacity:TOSEM:undecidabilityFullOpacityPTA}
\end{figure}
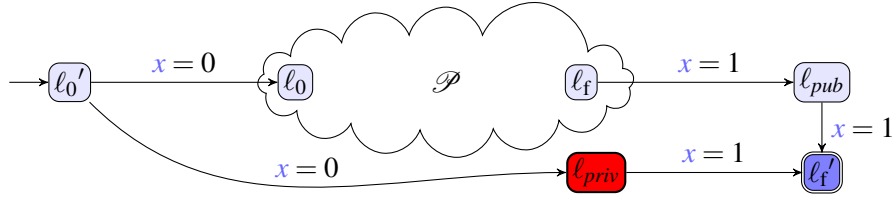
\begin{proof}
We reduce from the reachability-emptiness problem in bounded time, which is undecidable from~\cite[Theorem~3.12]{ALM20}.
(This is different from the proof of \cite[Theorem~7.2]{ALMS22}, which reduces from the reachability-emptiness problem in \emph{constant} time, which is undecidable according to \cite[Lemma~7.1]{ALMS22}.)

Consider an arbitrary PTA~$\PTA$%
, with initial location~$\locinit$ and a final location~$\locfinal$.
We add the following locations and transitions in~$\PTA$\ to obtain a PTA~$\PTA'$, as in \cref{figure:opacity:TOSEM:undecidabilityFullOpacityPTA}:
\begin{inlineenumeration}
	\item a new
	initial location $\locinit'$, with outgoing transitions in 0-time (due to their guard $\clock = 0$, where~$\clock$ is a new clock not belonging to~$\PTA$, and never reset in~$\PTA'$) to~$\locinit$ and to a new location~$\locpriv$,
	\item a new location $\locpub$ with an incoming transition from~$\locfinal$ guarded by $\clock = 1$,
	and
	\item a new final location $\locfinal'$ with incoming transitions from $\locpub$\ and~$\locpriv$ both guarded by $\clock = 1$.
\end{inlineenumeration}

First, note that, due to the guarded transitions, $\locfinal'$ is reachable for any parameter valuation via runs visiting~$\locpriv$, (only) for an \execTimeText{} equal to~$1$.
That is, for all $\pval$, $\PrivDurVisit{\valuate{\PTA'}{\pval}} = \set{1}$.

We now show that there exists a valuation~$\pval$ such that $\valuate{\PTA'}{\pval}$ is \weakOpaqueText{} (with $\locpriv$ as private location, and $\locfinal'$ as final location)
iff
there exists a valuation~$\pval$ such that $\locfinal$ is reachable in $\valuate{\PTA}{\pval}$ for an \execTimeText{} $\leq 1$.
\begin{itemize}
\item[$\Leftarrow$]
Assume there exists some valuation~$\pval$ such that $\locfinal$ is reachable from $\locinit$ in~$\PTA$ for an \execTimeText{} $\leq 1$.
Then, due to our construction, $\locpub$ is visited on the way to $\locfinal'$ in $\valuate{\PTA'}{\pval}$ (only) for the \execTimeText{}  $1$.
Therefore,
$\PubDurVisit{\valuate{\PTA'}{\pval}} = \set{1} = \PrivDurVisit{\valuate{\PTA'}{\pval}}$
and then
$\valuate{\PTA'}{\pval}$ is \weakOpaqueText{} (and also \fullOpaqueText{}, which plays no role here).

\item[$\Rightarrow$]
Conversely, if $\locfinal$ is not reachable from $\locinit$ in $\PTA$ for any valuation for an \execTimeText{} $\leq 1$, then no run reaches $\locfinal'$ in time~1 without visiting~$\locpriv$, for any valuation of $\PTA'$.
Therefore, for any valuation~$\pval$,
$\PrivDurVisit{\valuate{\PTA'}{\pval}} = \set{1} \not\subseteq \PubDurVisit{\valuate{\PTA'}{\pval}}=\emptyset$.
Therefore, there is no valuation~$\pval$ such that $\valuate{\PTA'}{\pval}$ is \weakOpaqueText{}.
\end{itemize}
Therefore, there exists a valuation~$\pval$ such that $\valuate{\PTA'}{\pval}$ is \weakOpaqueText{}
iff
there exists a valuation~$\pval$ such that $\locfinal$ is reachable in $\valuate{\PTA}{\pval}$ for an \execTimeText{} $\leq 1$---which is undecidable from~\cite[Theorem~3.12]{ALM20}.
This concludes the proof.
\end{proof}

Since the emptiness problem is undecidable, the synthesis problem is immediately unsolvable as well.

\begin{corollary}\label{prop:opacity:TOSEM:W-ET-psynthesis-PTA}
	The \weakOpacityParamSynthesisProblem{} is unsolvable for general \PTAstext{}.
\end{corollary}
\subsubsection{Undecidability for \LUPTAstextForSec{}}

We provide below another original result in the context of \emph{weak} opacity, this time for \LUPTAstext{}, largely inspired by the proof of \cref{prop:opacity:TOSEM:F-ET-pemptiness-LUPTA}, even though our construction needed to be changed.

\begin{theorem}[Undecidability of the \weakOpacityParamEmptinessProblem{} for \LUPTAstext{}]\label{prop:opacity:TOSEM:W-ET-pemptiness-LUPTA}
	The \weakOpacityParamEmptinessProblem{} is undecidable for \LUPTAstext{}.
\end{theorem}
\begin{figure}[tb]
{\centering
\resizebox{\textwidth}{!}{
	\begin{tikzpicture}[->, >=stealth', auto, node distance=2cm, scale=1.5, thin]

		\node[location] (l0) at (0, 0) {$\locinit$};
		\node[cloud, cloud puffs=15.7, cloud ignores aspect, minimum width=4cm, minimum height=2cm, align=center, draw] (cloud) at (.75, 0) {$\PTA$};
		\node[location] (lf) at (+1.5, 0) {$\locfinal$};

		\node[location, initial] (l0') at (-7, 0) {$\locinit'$};
		\node[location] (l1) at (-4.75, .7) {$\loc_1$};
		\node[location] (l2) at (-4.75, -.7) {$\loc_2$};
		\node[location] (l3) at (-2.5, 0) {$\loc_3$};
		\node[location, private] (lpriv) at (+1.8, -1.2) {$\locpriv$};
		\node[location] (l4) at (-1, -1.2) {$\loc_4$};

		\node[location, final] (lf') at (+5, 0) {$\locfinal'$};

		\path
		(lf) edge node[] {$\sclock = 2$} (lf')
		(l0') edge node[sloped,align=center] {$\styleparam{\paramiLower{1}} \leq \sclock \leq \styleparam{\paramiUpper{1}})$} (l1)
		(l1) edge node[sloped,align=center] {$\styleparam{\paramiLower{2}} \leq \sclock \leq \styleparam{\paramiUpper{2}})$} (l3)
		(l0') edge node[sloped,align=center,below] {$\styleparam{\paramiLower{2}} \leq \sclock \leq \styleparam{\paramiUpper{2}})$} (l2)
		(l2) edge node[sloped,align=center,below] {$\styleparam{\paramiLower{1}} \leq \sclock \leq \styleparam{\paramiUpper{1}})$} (l3)
		(l3) edge node[above,align=center,xshift=-1em] {$\sclock = 1$} node[below] {$\styleclock{\ClockSet} \setminus \{ \sclock \} \assign 0$} (l0)

		(l0') edge[out=300,in=180] node[below right] {$\bigvee_{i}(\styleparam{\paramiLower{i}} < \sclock \leq \styleparam{\paramiUpper{i}})$} (l4)
		(l4) edge[] node[above] {$\sclock > 2$} node[below] {$\sclock \assign 0$} (lpriv)

		(l0') edge[out=-60,in=-155] node[above] {$\sclock = 2$} node[below] {$\sclock \assign 0$} (lpriv)

		(lpriv) edge[bend right] node[] {$\sclock = 0$} (lf')

		;

	\end{tikzpicture}
}
}
\caption{Undecidability of \fullOpacityParamEmptinessProblem{} for \LUPTAstextForFig{}}
\label{figure:opacity:TOSEM:undecidabilityFTOELU}
\end{figure}
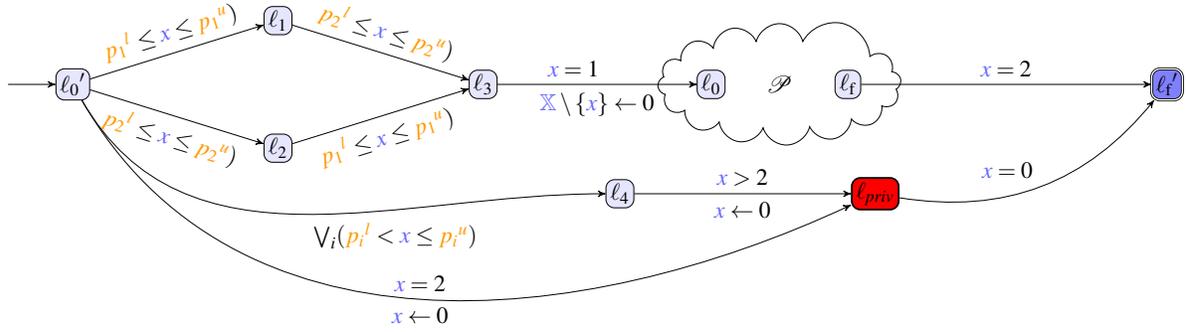
\begin{proof}
Let us recall from~\cite[Theorem~3.12]{ALM20} that the reachability-emptiness problem is undecidable over bounded time for \PTAstext{} with (at least) 3~clocks and 2~parameters.
Assume a \PTAtext{}~$\PTA$ with 3~clocks and 2~parameters, say $\parami{1}$ and $\parami{2}$, and a final location~$\locfinal$.
Take~$1$ as a time bound.
From \cite[Theorem~3.12]{ALM20}, it is undecidable whether there exists a parameter valuation for which~$\locfinal$ is reachable in~$\PTA$ in time $\leq 1$.

The idea of our proof is that, as in~\cite{JLR15,ALMS22}, we ``split'' each of the two parameters used in~$\PTA$ into a lower-bound parameter ($\paramiLower{1}$ and $\paramiLower{2}$) and an upper-bound parameter ($\paramiUpper{1}$ and $\paramiUpper{2}$).
Each constraint of the form $\clock < \parami{i}$ (resp.\ $\clock \leq \parami{i}$) is replaced with $\clock < \paramiUpper{i}$ (resp.\ $\clock \leq \paramiUpper{i}$)
while
each constraint of the form $\clock > \parami{i}$ (resp.\ $\clock \geq \parami{i}$) is replaced with $\clock > \paramiLower{i}$ (resp.\ $\clock \geq \paramiLower{i}$);
$\clock = \parami{i}$ is replaced with $\paramiLower{i} \leq \clock \leq \paramiUpper{i}$.

The idea is that the PTA~$\PTA$ is exactly equivalent to our construction with duplicated parameters only when $\paramiLower{1} = \paramiUpper{1}$ and $\paramiLower{2} = \paramiUpper{2}$.
The crux of the rest of this proof is that we will ``rule out'' any parameter valuation not satisfying these equalities, so as to use directly the undecidability result of~\cite[Theorem~3.12]{ALM20}.

Now, consider the extension of~$\PTA$ given in \cref{figure:opacity:TOSEM:undecidabilityFTOELU}, and let~$\PTA'$ be this extension.
We assume that $\clock$ is an extra clock not used in~$\PTA$.
The syntax ``$\ClockSet \setminus \{ \clock \} \assign 0$'' denotes that all clocks of the original \PTAtext{}~$\PTA$ are reset---but not the new clock~$\clock$.
The guard on the transition from $\locinit'$ to $\loc_4$ stands for 2 different transitions guarded with
$\paramiLower{1} < \clock \leq \paramiUpper{1}$,
and
$\paramiLower{2} < \clock \leq \paramiUpper{2}$,
respectively.

Let us first make the following observations:
\begin{enumerate}
	\item for any parameter valuation, one can take the transition from $\locinit'$ to~$\locpriv$ at time~$2$ and then to~$\locfinal'$ in 0-time (\ie{} at time~2), \ie{} $\locfinal'$ is always reachable in time~$2$ while visiting location~$\locpriv$; put differently, $\{ 2 \} \subseteq \PrivDurVisit{\valuate{\PTA'}{\pval}}$ for any parameter valuation~$\pval$;

	\item the original automaton~$\PTA$ can only be entered whenever $\paramiLower{1} \leq \paramiUpper{1}$ and $\paramiLower{2} \leq \paramiUpper{2}$; going from~$\locinit'$ to~$\locinit$ takes exactly $1$ time unit (due to the $\clock = 1$ guard);

	\item to reach~$\locfinal'$ without visiting~$\locpriv$, a run must go through~$\PTA$ and visit~$\locfinal$, and its duration is necessarily~$2$;
		put differently, $\PubDurVisit{\valuate{\PTA'}{\pval}} \subseteq \{ 2 \}$ for any parameter valuation~$\pval$;

	\item from~\cite[Theorem~3.12]{ALM20}, it is undecidable whether there exists a parameter valuation for which there exists a run reaching~$\locfinal$ from~$\locinit$ in time $\leq 1$, \ie{} reaching~$\locfinal$ from~$\locinit'$ in time $\leq 2$.
\end{enumerate}

Let us consider the following cases depending on the valuations:
\begin{enumerate}
	\item for valuations~$\pval$ such that $\paramiLower{1} > \paramiUpper{1}$ or $\paramiLower{2} > \paramiUpper{2}$, then thanks to the transitions from $\locinit'$ to~$\locinit$, there is no way to enter the original \PTAtext{}~$\PTA$ (and therefore to reach $\locfinal'$ without visiting~$\locpriv$);
	hence, $\PubDurVisit{\valuate{\PTA'}{\pval}} = \emptyset$, and therefore $\{ 2 \} \subseteq \PrivDurVisit{\valuate{\PTA'}{\pval}} \not\subseteq \PubDurVisit{\valuate{\PTA'}{\pval}}$, \ie{} $\PTA'$ is not \weakOpaqueText{} for any of these valuations.

	\item for valuations~$\pval$ such that $\paramiLower{1} < \paramiUpper{1}$ or $\paramiLower{2} < \paramiUpper{2}$, then the transition from~$\locinit'$ to~$\loc_4$ can be taken, and therefore there exist runs reaching $\locfinal'$ after a duration $> 2$ (for example of duration~3) and visiting~$\locpriv$.
	Since no run can reach $\locfinal'$ without visiting~$\locpriv$ for a duration $\neq 2$, then
	$\{ 3 \} \subseteq \PrivDurVisit{\valuate{\PTA'}{\pval}} \not\subseteq \PubDurVisit{\valuate{\PTA'}{\pval}} \subseteq \{ 2 \}$
	and
	again $\PTA'$ is not \weakOpaqueText{} for any of these valuations.

	\item for valuations such that $\paramiLower{1} = \paramiUpper{1}$ and $\paramiLower{2} = \paramiUpper{2}$, then the behavior of the modified~$\PTA$ (with duplicate parameters) is exactly the one of the original~$\PTA$.
	Also, note that the transition from~$\locinit'$ to~$\loc_4$ cannot be taken.
	In contrast, the transition from~$\locinit'$ to~$\locpriv$ can still be taken, and therefore there exists a run of duration~$2$ visiting~$\locpriv$ and reaching~$\locfinal'$.
	Hence, $\PrivDurVisit{\valuate{\PTA'}{\pval}} = \{ 2 \}$ for any such valuation~$\pval$.

	\begin{itemize}
		\item
		Now, assume there exists such a parameter valuation~$\pval$ for which there exists a run of~$\valuate{\PTA}{\pval}$ of duration $\leq 1$ reaching~$\locfinal$.
		And, as a consequence, there exists a run of~$\valuate{\PTA'}{\pval}$ of duration~$2$ (including the $1$ time unit to go from~$\locinit'$ to~$\locinit$) reaching $\locfinal'$ without visiting~$\locpriv$.
		Hence, $\PubDurVisit{\valuate{\PTA'}{\pval}} = \{ 2 \}$.
		Therefore $\PubDurVisit{\valuate{\PTA'}{\pval}} = \PrivDurVisit{\valuate{\PTA'}{\pval}} = \{ 2 \}$.

		As a consequence, the modified automaton $\PTA'$ is \weakOpaqueText{} (and actually \fullOpaqueText{}---which plays no role in this proof) for such a parameter valuation.

		\item Conversely, assume there exists no parameter valuation for which there exists a run of~$\PTA$ of duration $\leq 1$ reaching $\locfinal$.
            In that case, $\locfinal'$ can never be reached without visiting~$\locpriv$: $\PubDurVisit{\valuate{\PTA'}{\pval}} = \emptyset$, and therefore $\{ 2 \} \subseteq \PrivDurVisit{\valuate{\PTA'}{\pval}} \not\subseteq \PubDurVisit{\valuate{\PTA'}{\pval}}$, \ie{} $\valuate{\PTA'}{\pval}$ is not \fullOpaqueText{} for any such parameter valuation~$\pval$.

	\end{itemize}

\end{enumerate}

As a consequence, there exists a parameter valuation~$\pval'$ for which $\valuate{\PTA'}{\pval'}$ is \weakOpaqueText{} iff there exists a parameter valuation~$\pval$ for which there exists a run in~$\valuate{\PTA}{\pval}$ of duration $\leq 1$ reaching $\locfinal$---which is undecidable from~\cite[Theorem~3.12]{ALM20}.
\end{proof}
\begin{corollary}\label{prop:opacity:TOSEM:W-ET-psynthesis-LUPTA}
	The \weakOpacityParamSynthesisProblem{} is unsolvable for \LUPTAstext{}.
\end{corollary}
\section{Expiring execution-time opacity problems}\label{section:opacity:expiring}

In~\cite{AEYM21}, the authors consider a time-bounded notion of the opacity of~\cite{Cassez09}, where the attacker has to disclose the secret before an upper bound, using a partial observability.
This can be seen as a secrecy with an \emph{\expirationDateText{}}.
The rationale is that retrieving a secret ``too late'' is useless; this is understandable, \eg{} when the secret depends of the status of the memory; if the cache was overwritten since, then knowing the secret is probably useless in most situations.
In addition, the analysis in~\cite{AEYM21} is carried over a time-bounded horizon; this means there are two time bounds in~\cite{AEYM21}: one for the secret \expirationDateText{}, and one for the bounded-time execution of the system.

In this section, we review a recent work of ours~\cite{ALM23} in which we incorporate this secret \expirationDateText{} into our notion of \opacityText{}: we only consider the former notion of time bound from~\cite{AEYM21} (the secret \expirationDateText{}), and lift the assumption regarding the latter (the bounded-time execution of the system).
More precisely, we consider an \emph{expiring version of \opacityText{}}, where the secret is subject to an \expirationDateText{};
this can be seen as a combination of both concepts from~\cite{ALMS22} and~\cite{AEYM21}.
That is, we consider that an attack is successful only when the attacker can decide that the secret location was entered less than $\expiringBound$ time units before the system completion.
Conversely, if the attacker exhibits an \execTimeText{}~$\paramd$ for which it is certain that the secret location was visited, but this location was entered strictly more than~$\expiringBound$ time units prior to the system completion, then this attack is useless, and can be seen as a failed attack.
The system is therefore \emph{fully} \tempOpaqueText{} if the set of \execTimesText{} for which the private location was entered within $\expiringBound$ time units prior to system completion is exactly equal to the set of \execTimesText{} for which the private location was either not visited or entered $> \expiringBound$ time units prior to system completion.

In addition, when the former (secret) set of \execTimesText{} is \emph{included} into the latter (non-secret) set of times, we say that the system is \emph{weakly} \tempOpaqueText{}; this encodes situations when the attacker might be able to deduce that no secret location was visited, but is not able to confirm that the secret location \emph{was} indeed visited.

\begin{table}[tb]
	\centering
	\caption{Summary of the definitions for \opacityText{} and expiring \opacityText{}~\cite{ALMS22,ALM23}}
	\newcommand{\rowLenght}{.4\textwidth}
	\begin{tabular}{|l|p{\rowLenght}|p{.4\textwidth}|}
		\hline
		\cellHeader{}                       & \cellHeader{Secret runs}                                                                                                  & \cellHeader{Non-secret runs}                                                                      \\ \hline
		\opacityText{}                      & Runs visiting the private location ($=$ private runs)                                                                     & Runs not visiting the private location ($=$ public runs)                                          \\ \hline
		\multirow{2}{*}{\tempOpacityText{}} & \multirow{2}{\rowLenght}{Runs visiting the private location $\leq\expiringBound$ time units before the system completion} & (i) Runs not visiting the private location and                                                    \\
		                                    &                                                                                                                           & (ii) Runs visiting the private location $>\expiringBound$ time units before the system completion \\ \hline
	\end{tabular}

	\medskip

	\newcommand{\setsecret}{\ensuremath{\{\text{secret runs}\}}}
	\newcommand{\setnonsecret}{\ensuremath{\{\text{non-secret runs}\}}}
	\begin{tabular}{|l|l|}
		\hline
		\cellHeader{The system is}     & \cellHeader{if}                                 \\
		\rowHeader{} \scriptsize{(resp.\ expiring)} &                                                 \\ \hline
		\opaqueText{}                  & $ \setsecret \cap \setnonsecret \neq \emptyset$ \\
		\weakOpaqueText{}              & $ \setsecret \subseteq \setnonsecret$           \\
		\fullOpacityText{}             & $ \setsecret = \setnonsecret$                   \\ \hline
	\end{tabular}
	\label{tab:definitionRecall}
\end{table}

On the one hand, our attacker model is \emph{less powerful} than~\cite{AEYM21}, because our attacker has only access to the \execTimeText{} (and to the input model); in that sense, our attacker capability is identical to~\cite{ALMS22}.
On the other hand, we lift the time-bounded horizon analysis from~\cite{AEYM21}, allowing to analyze systems without any assumption on their \execTimeText{}; therefore, we only import from~\cite{AEYM21} the notion of \emph{expiring secret}.

We summarize in \cref{tab:definitionRecall} our different notions of \opacityText{} and expiring \opacityText{}; we will define formally  expiring \opacityText{} in the following.

\subsection{\TempOpacityTextForSec{}}

Let us first introduce some notions dedicated to expiring \opacityText{} (hereafter referred to as \tempOpacityText{}).
Let $\setRplusinf = \setRgeqzero \cup \{+\infty\}$.
Given a \TAtext{} $\TA$ and a finite run $\run$ in \semantics{\TA}, the \emph{duration} between two states of $\run : \concstatei{0}, (d_0, \edgei{0}), \concstatei{1}, \cdots, \concstatei{k}$ is $\statesDuration{\run}{\concstatei{i}}{\concstatei{j}} = \sum_{i \leq m \leq j-1} d_m$.
We also define the \emph{duration} between two locations $\loci{1}$ and $\loci{2}$ as the duration $\locationsDuration{\run}{\loci{1}}{\loci{2}} = \statesDuration{\run}{\concstatei{i}}{\concstatei{j}}$ with
$\run : \concstatei{0}, (d_0, \edgei{0}), \concstatei{1}, \cdots, \concstatei{i}, \cdots, \concstatei{j}, \cdots, \concstatei{k}$ where
$\concstatei{j}$ the first occurrence of a state with location $\loci{2}$
and
$\concstatei{i}$ is the last state of~$\run$ with location $\loci{1}$ before~$\concstatei{j}$.
We choose this definition to coincide with the definitions of opacity that we will define in the following \cref{def:opacity:ICECCS:temporary-timed-opacity}. Indeed, we want to make sure that revealing a secret ($\loci{1}$ in this definition) is not a failure if it is done after a given time. Thus, as soon as the system reaches its final state ($\loc_2$), we will be interested in knowing how long the secret has been present, and thus the last time it was entered~($\concstatei{i}$).

Given $\expiringBound \in \setRplusinf$, we define $\TempInfPrivVisit{\TA}{\expiringBound}$ (resp.\ $\TempSupPrivVisit{\TA}{\expiringBound}$)
as the set of runs $\run\in\PrivVisit{\TA}$
\st\ $\locationsDuration{\run}{\locpriv}{\locfinal} \leq \expiringBound$ (resp.\ $\locationsDuration{\run}{\locpriv}{\locfinal} > \expiringBound$).
We refer to the runs of $\TempInfPrivVisit{\TA}{\expiringBound}$ as \emph{secret} runs; their durations are denoted by $\TempInfPrivDurVisit{\TA}{\expiringBound}$.
Similarly, the durations of the runs of $\TempSupPrivVisit{\TA}{\expiringBound}$ are denoted by $\TempSupPrivDurVisit{\TA}{\expiringBound}$.

We define below two notions of \opacityText{} \wrt\ a time bound~$\expiringBound$.
We will compare two sets:
\begin{enumerate}%
	\item the set of \execTimesText{} for which the private location was entered at most~$\expiringBound$ time units prior to system completion; and
	\item the set of \execTimesText{} for which either the private location was not visited at all, or it was last entered more than~$\expiringBound$
	time units prior to system completion (which, in our setting, is somehow similar to \emph{not} visiting the private location, in the sense that entering it ``too early'' is considered of little interest).
\end{enumerate}%
If both sets match, the system is \fullTempOpaqueTextVal{\expiringBound}.
If the former is included into the latter, then the system is \weakTempOpaqueTextVal{\expiringBound}.

\begin{definition}[\Acl*{tempopacity}]\label{def:opacity:ICECCS:temporary-timed-opacity}
	Given a \TAtext{}~$\TA$ and
	a bound (\ie{} an \expirationDateText{} for the secret)~$\expiringBound\in\setRplusinf$
	we say that $\TA$ is \emph{\fullTempOpaqueText{}} \wrt\ the \expirationDateText{} \expiringBound, denoted by \emph{\fullTempOpaqueTextVal{\expiringBound}}, if
	$$\TempInfPrivDurVisit{\TA}{\expiringBound} = \TempSupPrivDurVisit{\TA}{\expiringBound} \cup \PubDurVisit{\TA}.$$

	Moreover,
	$\TA$ is \emph{\weakTempOpaqueText{}} \wrt\ the \expirationDateText{} \expiringBound, denoted by \emph{\weakTempOpaqueTextVal{\expiringBound}}, if
	$$\TempInfPrivDurVisit{\TA}{\expiringBound} \subseteq \TempSupPrivDurVisit{\TA}{\expiringBound} \cup \PubDurVisit{\TA}.$$
	
	Finally, 	
	$\TA$ is \emph{\existentialOpaqueText{}} \wrt\ the \expirationDateText{} \expiringBound, denoted by \emph{\existentialTempOpaqueTextVal{\expiringBound}},
	if $$\TempInfPrivDurVisit{\TA}{\expiringBound} \cap (\TempSupPrivDurVisit{\TA}{\expiringBound} \cup \PubDurVisit{\TA}) \neq \emptyset.$$

\end{definition}
\begin{example}\label{ex:opacity:ICECCS:defOpacity}
	Consider again the \PTAtext{} in \cref{figure:example-PTA}; let $\pval$ be such that $\pval(\parami{1}) = 1$ and $\pval(\parami{2}) = 2.5$.
	Fix $\expiringBound = 1$.

	We have:
	\begin{itemize}
		\item $\PubDurVisit{\valuate{\PTA}{\pval}} = [0, 3] $
		\item $\TempSupPrivDurVisit{\valuate{\PTA}{\pval}}{\expiringBound} = (2,2.5] $
		\item $\TempInfPrivDurVisit{\valuate{\PTA}{\pval}}{\expiringBound} = [1,2.5] $
	\end{itemize}

	Therefore, we say that $\valuate{\PTA}{\pval}$ is: %
	\begin{itemize}
		\item \existentialTempOpaqueTextVal{1}, as $[1,2.5] \cap \big( (2,2.5] \cup [0,3] \big) \neq \emptyset$
		\item \weakTempOpaqueTextVal{1}, as $[1,2.5] \subseteq \big( (2,2.5] \cup [0,3] \big)$
		\item not \fullTempOpaqueTextVal{1}, as $[1,2.5] \neq \big( (2,2.5] \cup [0,3] \big)$
	\end{itemize}

	As noted in \cref{remark:opacity:ICECCS:weak}, despite the \weakTempOpacityTextVal{1} of~$\TA$, the attacker can deduce some information about the visit of the private location for some \execTimesText{}.
	For example, if a run has a duration of 3 time units, it cannot be a private run, and therefore the attacker can deduce that the private location was not visited at all.
\end{example}
\subsection{\TempOpacityTextForSec{} problems in timed automata}
\subsubsection{Problem definitions}

We define seven different problems in the context of (non-parametric) \TAstext{}:

\defProblem
{\existentialTempDecisionProblem{}}
{A \TAtext{}~$\TA$ and a bound $\expiringBound\in\setRplusinf$}
{Decide whether $\TA$ is \existentialTempOpaqueTextVal{\expiringBound}.}%
\defProblem
{\WeakfullTempDecisionProblem{}}
{A \TAtext{}~$\TA$ and a bound $\expiringBound\in\setRplusinf$}
{Decide whether $\TA$ is \weakfullTempOpaqueTextVal{\expiringBound}.}%
\defProblem
{\WeakfullTempBoundEmptinessProblem{}}
{A \TAtext{}~$\TA$}
{Decide the emptiness of the set of bounds $\expiringBound$ such that $\TA$ is \weakfullTempOpaqueTextVal{\expiringBound}.}
\defProblem
{\WeakfullTempBoundComputationProblem{}}
{A \TAtext{}~$\TA$}
{Compute the maximal set $\setBound$ of bounds such that $\TA$ is \weakfullTempOpaqueTextVal{\expiringBound} for all $\expiringBound\in\setBound$.}
\begin{example}\label{ex:opacity:ICECCS:opacityProblemsTA}
	Consider again the \PTAtext{} in \cref{figure:example-PTA}; let $\pval$ be such that $\pval(\parami{1}) = 1$ and $\pval(\parami{2}) = 2.5$ (as in \cref{ex:opacity:ICECCS:defOpacity}).
	Let us exemplify some of the problems defined above.

	\begin{itemize}
		\item Given $\expiringBound = 1$, the \weakTempDecisionProblem{} asks whether $\valuate{\PTA}{\pval}$ is \weakTempOpaqueTextVal{1}---the answer is ``yes'' from \cref{ex:opacity:ICECCS:defOpacity}.

		\item The answer to the \weakTempBoundEmptinessProblem{} is therefore ``no'' because the set of bounds $\expiringBound$ such that $\valuate{\PTA}{\pval}$ is \weakTempOpaqueTextVal{\expiringBound} is not empty.

		\item Finally, the \weakTempBoundComputationProblem{} asks to compute all the corresponding bounds: in this example, the solution is ${\expiringBound\in\setRplusinf}$, \ie{} the solution is the set all possible (non-negative) values for~$\expiringBound$. %
	\end{itemize}
\end{example}
\paragraph{Relations with the \opacityTextForSec{} problems}
Note that, when considering $\expiringBound = +\infty$,
$\TempSupPrivDurVisit{\TA}{\expiringBound}=\emptyset$ and all the \execTimesText{} of runs visiting $\locpriv$ are in $\TempInfPrivDurVisit{\TA}{\expiringBound}$.
Therefore, \fullTempOpacityTextVal{\ensuremath{+\infty}} matches the \fullOpacityText{}.
We can therefore notice that answering the \fullTempDecisionProblem{} for $\expiringBound=+\infty$ is decidable (\cref{prop:opacity:TOSEM:F-ET-decision}).
However, the emptiness and computation problems cannot be reduced to \fullOpacityText{} problems from \cref{ss:ET-opacity-LU}.

Conversely, it is possible to answer the \fullOpacityDecisionProblem{} by checking the \fullTempDecisionProblem{} with $\expiringBound=+\infty$.
Moreover, the \opacityExecutionTimesComputationProblem{}
reduces to the \fullTempBoundComputationProblem{}: if $+\infty \in \setBound$,
we get the answer.

Recall that we summarize our different definitions of (expiring) \opacityText{} in \cref{tab:definitionRecall}.

\subsubsection{Results}
In general, the link between the full and weak notions of the three aforementioned problems is not obvious.
However, for a fixed value of $\expiringBound$, we establish the following theorem.

\begin{theorem}[%
{\cite[Theorem~1]{ALM23}}]\label{thm:opacity:ICECCS:reduction-weak-full}
	The \fullTempDecisionProblem{} reduces to the \weakTempDecisionProblem{}.
\end{theorem}

We can now study the aforementioned problems.

\begin{theorem}[Decidability of \weakfullTempDecisionProblem{}~{\cite[Theorems~2 and~5]{ALM23}}]
	\label{thm:opacity:ICECCS:TAtempop}
	The \weakfullTempDecisionProblem{} is decidable in \NEXPTIME.
\end{theorem}
\begin{remark}\label{remark:ICECCS:expiring-extend-opacity}
	In \cref{prop:opacity:TOSEM:F-ET-decision}, we established that the \decisionProblem{\fullTempOpacityTextVal{\ensuremath{+\infty}}} is in \FiveEXPTIME.
	\cref{thm:opacity:ICECCS:TAtempop} thus extends our former results in three ways:
	\begin{enumerate}
		\item by including the parameter $\expiringBound$,
		\item by reducing the complexity and
		\item by considering as well the \emph{weak} notion of ET-opacity (considered separately in \cref{prop:opacity:TOSEM:W-ET-decision}).
	\end{enumerate}
\end{remark}

We complete these results from~\cite{ALM23} with the following result analog to \cref{prop:opacity:TOSEM:E-ET-decision}.
\begin{theorem}[Decidability of \existentialTempDecisionProblem{}]
	\label{thm:opacity:exist:TAtempop}
	The \existentialTempDecisionProblem{} is decidable in \PSPACE.
\end{theorem}
\begin{proof}
The \weakfullTempDecisionProblem{} was solved in~\cite{ALM23} by building two non-deterministic finite automata whose languages represented the secret and the non-secret durations of the system, respectively.
These automata being of exponential size and with a unary language, testing the equality or inclusion of languages led to the \NEXPTIME algorithm
quoted in \cref{thm:opacity:ICECCS:TAtempop}.
Similarly, the \existentialTempDecisionProblem{} can be decided by testing whether the intersection of the languages of these automata is empty.
This can be done in \NLOGSPACE in the size of the automata (classically, by first building the product between these two automata, and then by checking the reachability of a pair of final states), hence the \PSPACE algorithm.
\end{proof}
\begin{theorem}[Solvability of \weakTempBoundComputationProblem{}~{\cite[Theorems~3 and~5]{ALM23}}]
	\label{thm:opacity:ICECCS:weakcomp}
	The \weakTempBoundComputationProblem{} is solvable.
\end{theorem}
\begin{corollary}[Decidability of \weakTempBoundEmptinessProblem{}~{\cite[Corollary~1]{ALM23}}]\label{thm:opacity:ICECCS:weakEmptiness}
	The \weakTempBoundEmptinessProblem{} is decidable.
\end{corollary}

In contrast to the \weakTempBoundComputationProblem{}, we only show below that the \fullTempBoundEmptinessProblem{} is decidable; the computation problem remains open.

\begin{theorem}[Decidability of the \fullTempBoundEmptinessProblem{}~{\cite[Theorems~4 and~5]{ALM23}}]\label{thm:opacity:ICECCS:opacityEmptiness}
	The \fullTempBoundEmptinessProblem{} is decidable.
\end{theorem}
\subsection{\TempOpacityTextForSec{} in \PTAstextForSec{}}\label{sec:opacity:ICECCS:PTAs}

We now study \tempOpacityText{} problems for \PTAstext{}:
we will be interested in the synthesis and in the emptiness of the valuations set ensuring that a system is \weakfullTempOpaqueText{}.

\subsubsection{Definitions}

We define the following problems, where we ask for parameter valuations $\pval$ and for valuations of~$\expiringBound$ \st\ $\valuate{\PTA}{\pval}$ is \weakfullTempOpaqueTextVal{\expiringBound}.

\smallskip

\defProblem
{\WeakfullTempBoundParametersEmptinessProblem{}}
{A \PTAtext{}~$\PTA$
}
{Decide whether the set of parameter valuations~$\pval$ and valuations of~$\expiringBound$ such that $\valuate{\PTA}{\pval}$ is \weakfullTempOpaqueTextVal{\expiringBound} is empty}

\smallskip

\defProblem
{\WeakfullTempBoundParametersSynthesisProblem{}}
{A \PTAtext{}~$\PTA$
}
{Synthesize the set of parameter valuations~$\pval$ and valuations of~$\expiringBound$ such that $\valuate{\PTA}{\pval}$ is \weakfullTempOpaqueTextVal{\expiringBound}}

\begin{example}\label{ex:opacity:ICECCS:opacityProblemsPTA}
	Consider again the \PTAtext{}~$\PTA$ in \cref{figure:example-PTA}.

	For this \PTAtext{}, the answer to the \weakTempBoundParametersEmptinessProblem{} is false,
	as there exists such a valuation (\eg{} the valuation given in \cref{ex:opacity:ICECCS:opacityProblemsTA}).

	Moreover, we can show that, for all $\expiringBound$ and~$\pval$:
	\begin{itemize}
		\item $\PubDurVisit{\valuate{\PTA}{\pval}} = [0,3] $
		\item if $\valuate{\parami{1}}{\pval}>3$ or $\valuate{\parami{1}}{\pval}>\valuate{\parami{2}}{\pval}$, it is not possible to reach $\locfinal$ with a run visiting $\locpriv$ and therefore $\TempSupPrivDurVisit{\valuate{\PTA}{\pval}}{\expiringBound} = \TempInfPrivDurVisit{\valuate{\PTA}{\pval}}{\expiringBound} = \emptyset$
		\item if $\valuate{\parami{1}}{\pval}\leq3$ and $\valuate{\parami{1}}{\pval}\leq\valuate{\parami{2}}{\pval}$
		\begin{itemize}
			\item $\TempSupPrivDurVisit{\valuate{\PTA}{\pval}}{\expiringBound} = (\valuate{\parami{1}}{\pval} + \expiringBound , \valuate{\parami{2}}{\pval}] $
			\item $\TempInfPrivDurVisit{\valuate{\PTA}{\pval}}{\expiringBound} = [\valuate{\parami{1}}{\pval}, \min(\expiringBound+3, \valuate{\parami{2}}{\pval})] $
		\end{itemize}
	\end{itemize}

	Recall that the \fullTempBoundParametersSynthesisProblem{} aims at synthesizing the valuations such that
	$\TempInfPrivDurVisit{\valuate{\PTA}{\pval}}{\expiringBound} = \TempSupPrivDurVisit{\valuate{\PTA}{\pval}}{\expiringBound} \cup \PubDurVisit{\valuate{\PTA}{\pval}}$.
	The answer to this problem is therefore the set of valuations of timing parameters and of~$\expiringBound$ \st:
	$$\pval(\parami{1})=0 \wedge \Big( \big(\expiringBound\leq 3 \wedge 3\leq \pval(\parami{2})\leq \expiringBound+3\big)\vee \big(\pval(\parami{2})<\expiringBound\wedge \pval(\parami{2})=3\big)\Big)\text{.}$$
\end{example}
\subsubsection{Results}
\paragraph{The subclass of \LUPTAstextForSec{}}
\begin{theorem}[Undecidability of \weakfullTempBoundParametersEmptinessProblem{}~{\cite[Theorem~6]{ALM23}}]\label{thm:opacity:ICECCS:EmptinessLU}
	The \weakfullTempBoundParametersEmptinessProblem{} is undecidable for \LUPTAstext{}. %
\end{theorem}
The synthesis problems are therefore immediately unsolvable as well.

\begin{corollary}[{\cite[Corollary~2]{ALM23}}]\label{thm:opacity:ICECCS:synthesisLU}
	The \weakfullTempBoundParametersSynthesisProblem{} is unsolvable for \LUPTAstext{}. %
\end{corollary}
\paragraph{The full class of \PTAstextForSec{}}

The undecidability of the emptiness problems for \LUPTAstext{} proved above (\cref{thm:opacity:ICECCS:EmptinessLU}) immediately implies undecidability for the larger class of \PTAstext{}.
However, as in \cref{remark:redundancy}, the full proof (given in~\cite{ALM23}) of the result stated below uses less clocks and parameters than for \LUPTAstext{} (\cref{thm:opacity:ICECCS:EmptinessLU}).

\begin{theorem}[Undecidability of \weakfullTempBoundParametersEmptinessProblem{}~{\cite[Theorem~7]{ALM23}}]\label{thm:opacity:ICECCS:EmptinessPTA}
	The \weakfullTempBoundParametersEmptinessProblem{} is undecidable for general \PTAstext{}. %
\end{theorem}

Again, the synthesis problems are therefore immediately unsolvable as well.

\begin{corollary}[{\cite[Corollary~3]{ALM23}}]\label{thm:opacity:ICECCS:synthesisPTA}
	The \weakfullTempBoundParametersSynthesisProblem{} is unsolvable for \PTAstext{}. %
\end{corollary}
\section{Implementation and application to Java programs}\label{section:implementation}

A motivation for the works on \opacityText{} (described in \cref{section:opacity:TA,section:opacity:PTA}) is the analysis of programs.
More precisely, we are interested in deciding whether a program, \eg{} written in Java, is \opaqueText{}, \ie{} whether an attacker is incapable of deducing internal behavior by only looking at its \execTimeText{}.
A second motivation is the \emph{configuration} of internal timing values from a program, \eg{} changing some internal delays, or tuning some \stylecode{Thread.sleep()} statements in the program, so that the program becomes \opaqueText{}---justifying notably the results in \cref{section:opacity:PTA}.

\paragraph{Semi-algorithm and implementation}
Despite the negative theoretical results (notably \cref{prop:opacity:TOSEM:E-ET-pemptiness-PTA}), we addressed in~\cite{ALMS22} the \existentialOpacityParamSynthesisProblem{} for the full class of \PTAstext{}.
Our method may not terminate (due to the undecidability) but, if it does, its result is correct.
Our workflow~\cite{ALMS22} can be summarized as follows.

\begin{enumerate}
	\item We slightly modify the original \PTAtext{} (by adding a Boolean flag~$\bflag$ and a final synchronization action);
	\item We perform \emph{self-composition} (\ie{} parallel composition with a copy of itself) of this modified \PTAtext{}, a method commonly used in security analyses~\cite{TA05,BDR11};
	\item We perform reachability-synthesis (\ie{} the synthesis of parameter valuations for which a given location is reachable) on~$\locfinal$ with contradictory values of~$\bflag$.
\end{enumerate}

Reachability-synthesis is implemented in \imitator{}~\cite{Andre21}, a parametric timed model checker taking as inputs networks of (extensions of) parametric timed automata, and synthesizing parameter valuations for which a number of properties (including reachability) hold.

\paragraph{Analysis of Java programs}
In addition, we are interested in analyzing programs too.
In order to apply our method to the analysis of programs, we need a systematic way of translating a program (\eg{} a Java program) into a \PTAtext{}.
In general, precisely modeling the \execTimeText{}  of a program using models like \TAtext{} is highly non-trivial due to complication of hardware pipelining, caching, OS scheduling, etc.
The readers are referred to the rich literature in, \eg{}~\cite{LYGY10,CB13}.
In~\cite{ALMS22}, we instead make the following simplistic assumption on \execTimeText{} of a program statement and focus on solving the parameter synthesis problem.
We assume that the \execTimeText{}  of a program statement other than \stylecode{Thread.sleep(n)} is within a range $[0,\epsilon]$ where $\epsilon$ is a small integer constant (in milliseconds), whereas the \execTimeText{}  of statement \stylecode{Thread.sleep(n)} is within a range $[n , n+\epsilon]$.
In fact, we choose to keep $\epsilon$ \emph{parametric} to be as general as possible, and to not depend on particular architectures.

Our test subject is a set of benchmark programs from the DARPA Space/Time Analysis for Cybersecurity (STAC) program.\footnote{\url{https://github.com/Apogee-Research/STAC/}}
These programs are being released publicly to facilitate researchers to develop methods and tools for identifying STAC vulnerabilities in the programs.
These programs are simple yet non-trivial, and were built on purpose to highlight vulnerabilities that can be easily missed by existing security analysis tools.
We \emph{manually} translated these programs to \PTAstext{}, following the method described above, and using a number of assumptions (such as collapsing loops with predefined duration).

In addition, we applied our method to a set of \PTAstext{} examples from the literature, notably from~\cite{HRSV02,GMR07,BCLR15,VNN18}.

Experiments reported in~\cite{ALMS22} show that we can decide whether these benchmarks (including the programs) are \fullOpaqueText{} or \existentialOpaqueText{}.
When adding timing parameters, we additionally answer the \existentialOpacityParamSynthesisProblem{}, \ie{} we synthesize the parameter valuations~$\pval$ and the associated \execTimesText{} \execTimes{} such that $\valuate{\PTA}{\pval}$ is \opaqueText{}.
Our method allows to exhibit cases when the system can never be made \opaqueText{}, including by tuning internal delays, or is always \opaqueText{}, or is \opaqueText{} only for some \execTimesText{} and internal timing parameters.

To summarize, the following problems can be answered using our framework:
\begin{itemize}
	\item \existentialOpacityDecisionProblem{}
	\item \fullOpacityDecisionProblem{}
	\item \weakOpacityDecisionProblem{} (not considered in our experiments in~\cite{ALMS22}, but can be easily adapted)
	\item \existentialOpacityParamSynthesisProblem{},  but without guarantee of termination, due to the undecidability of \cref{prop:opacity:TOSEM:E-ET-pemptiness-PTA}.
\end{itemize}

However, our procedure cannot in its current form answer neither the \fullOpacityParamSynthesisProblem{} nor the \weakOpacityParamSynthesisProblem{}.
The expiring opacity problems in \cref{section:opacity:expiring} were not addressed either.

\section{Conclusion and perspectives}\label{section:conclusion}

In this paper, we recalled (and proved a few original) results related to the \opacityText{} in \TAstext{}.
Our notion of \opacityText{} consists in considering an attacker model that can only observe the \execTimeText{} of the system, \ie{} the time from the initial location to a final location.
The secret consists in deciding whether a special private location was visited or not.
In contrast to another notion of opacity with a more powerful attacker able to observe some actions together with their timestamps, which led to the undecidability of the decision problem for \TAstext{}~\cite{Cassez09}, our notion of \opacityText{} yields decidability results for \TAstext{}.
Parameterizing the problems using timing parameters brings undecidability for \PTAstext{}, but the subclass of \LUPTAstext{} gives mildly positive results.

When in addition we consider that the secret has an expiration date, similarly to the concepts introduced in~\cite{AEYM21}, we are able to not only \emph{decide} problems for \TAstext{}, but also to \emph{synthesize} valuations for the expiration date such that the \TAtext{} is \weakTempOpaqueText{}.
However, problems extended with timing parameters all become undecidable.

Recall that we summarized in \cref{tab:resultsRecallTOSEM,tab:resultsRecallICECCS} the decidability results recalled in this paper, with a bold emphasis on the original results of this paper.

We also reported here on an implementation using \imitator{}, which is able to answer non-parametric problems (\existentialOpacityDecisionProblem{}, \fullOpacityDecisionProblem{}, \weakOpacityDecisionProblem{}), and also answering a parameter synthesis problem (\existentialOpacityParamSynthesisProblem) without guarantee of termination for the latter problem.

\paragraph{Perspectives}
The main theoretical future work is the open problems in \cref{tab:resultsRecallICECCS} (mainly the \fullTempBoundComputationProblem{}):
it is unclear whether we can \emph{compute} the exact set of \expirationDatesText{}~$\expiringBound$ for which a \TAtext{} is \fullTempOpaqueTextVal{\expiringBound}.

In terms of synthesis, we have so far no procedure able (whenever it terminates) to answer the \fullOpacityParamSynthesisProblem{} or the \weakOpacityParamSynthesisProblem{}.
Synthesis procedures to answer expiring opacity problems (defined in \cref{section:opacity:expiring}) for \PTAstext{} remain to be designed too.
These procedures cannot be both exact and guaranteed to terminate due to the aforementioned undecidability results.

Exact analysis of opacity for programs, including a more precise modeling of the cache, is also on our agenda, following works such as~\cite{CB13,CJM16}.

A different direction is that of \emph{control}:
can we turn a non-opaque system into an opaque system, by restraining its possible behaviors?
A first step with our notion of \opacityText{} was presented in~\cite{ABLM22}, with only an \emph{untimed} controller.
In addition, in~\cite{GMR07}, Gardey \etal\ propose several definitions of non-interference, related to various notions of simulation:
they consider not only the \emph{verification} problem (``is the system non-interferent?'')\ but also the (timed) \emph{control} problem (``synthesize a controller that will restrict the system in order to enforce non-interference'').
Extending our current line works on \opacityText{} to \emph{timed} controllers remains to be done.

\medskip

\paragraph{Acknowledgments}
We are grateful to Clemens Dubslaff and Maurice~ter~Beek for the opportunity to give an invited talk at TiCSA~2023, and for useful suggestions on this manuscript.

	\newcommand{\CCIS}{Communications in Computer and Information Science}
	\newcommand{\ENTCS}{Electronic Notes in Theoretical Computer Science}
	\newcommand{\FAC}{Formal Aspects of Computing}
	\newcommand{\FundInf}{Fundamenta Informaticae}
	\newcommand{\FMSD}{Formal Methods in System Design}
	\newcommand{\IJFCS}{International Journal of Foundations of Computer Science}
	\newcommand{\IJSSE}{International Journal of Secure Software Engineering}
	\newcommand{\IPL}{Information Processing Letters}
	\newcommand{\JAIR}{Journal of Artificial Intelligence Research}
	\newcommand{\JLAP}{Journal of Logic and Algebraic Programming}
	\newcommand{\JLAMP}{Journal of Logical and Algebraic Methods in Programming} %
	\newcommand{\JLC}{Journal of Logic and Computation}
	\newcommand{\LMCS}{Logical Methods in Computer Science}
	\newcommand{\LNCS}{Lecture Notes in Computer Science}
	\newcommand{\RESS}{Reliability Engineering \& System Safety}
	\newcommand{\RTS}{Real-Time Systems}
	\newcommand{\SCP}{Science of Computer Programming}
	\newcommand{\SOSYM}{Software and Systems Modeling ({SoSyM})}
	\newcommand{\STTT}{International Journal on Software Tools for Technology Transfer}
	\newcommand{\TCS}{Theoretical Computer Science}
	\newcommand{\TOPLAS}{{ACM} Transactions on Programming Languages and Systems ({ToPLAS})}
	\newcommand{\ToPNoC}{Transactions on {P}etri Nets and Other Models of Concurrency}
	\newcommand{\TOSEM}{{ACM} Transactions on Software Engineering and Methodology ({ToSEM})}
	\newcommand{\TSE}{{IEEE} Transactions on Software Engineering}

\bibliographystyle{eptcs}
\bibliography{PTA}

\end{document}